\documentclass[acmsmall,screen]{acmart}
\setcopyright{cc}
\acmDOI{10.1145/3729257}
\acmYear{2025}
\acmJournal{PACMPL}
\acmVolume{9}
\acmNumber{PLDI}
\acmArticle{158}
\acmMonth{6}
\received{2024-11-15}
\received[accepted]{2025-03-06}



\makeatletter
\def\@parfont{\bfseries}
\makeatother

\usepackage{local}
\usepackage{mathpartir}
\usepackage{proof-at-the-end}

\newif\ifproofatend%
\proofatendtrue% Comment out this line to remove proofs

\ifproofatend%
  \pgfkeys{/prAtEnd/default/.style={restate, proof at the end}}
  \pgfkeys{/prAtEnd/allend/.style={all end}}
\else
  \pgfkeys{/prAtEnd/default/.style={restate, proof at the end, stared, no link to proof}}
  \pgfkeys{/prAtEnd/allend/.style={all end, no link to theorem}}
\fi

\newif\ifcomments%
\ifcomments%
\newcommand{\jules}[1]{\textcolor{blue}{JJ:\@ #1}}
\newcommand{\alexandra}[1]{\textcolor{purple}{AS:\@ #1}}
\newcommand{\nate}[1]{\textcolor{green}{NF:\@ #1}}
\newcommand{\jana}[1]{\textcolor{red}{JW:\@ #1}}
\newcommand{\tobias}[1]{\textcolor{orange}{TK:\@ #1}}
\newcommand{\dexter}[1]{\textcolor{cyan}{DK:\@ #1}}
\else
\newcommand{\jules}[1]{}
\newcommand{\alexandra}[1]{}
\newcommand{\nate}[1]{}
\newcommand{\jana}[1]{}
\newcommand{\tobias}[1]{}
\newcommand{\dexter}[1]{}
\fi

\newcommand\StacKAT{\textsf{StacKAT}\xspace}

\usepackage{amsmath,tikz}
\usetikzlibrary{automata,backgrounds,shapes,arrows,positioning,calc,fit}
\usetikzlibrary{decorations,decorations.pathreplacing,decorations.pathmorphing,arrows.meta}

\newcommand\NF{\mathsf{nf}} % conflict with \nf
\newcommand\dn{{\downarrow}}
\newcommand\stackat{StacKAT}
\newcommand\SemR[1]{\ensuremath{R(#1)}}
\newcommand\Ssd{\Sigma^{*\dagger}}
\newcommand\set[2]{\{#1\mid#2\}}
\newcommand\eps\varepsilon%
\newcommand\BB{\mathcal{B}}

\newcommand\subs\subseteq%

\newcommand\Imp\Rightarrow%
\newcommand\Iff\Leftrightarrow%

\renewcommand\bar\overline%

\title{\StacKAT: Infinite State Network Verification}

\author{Jules Jacobs}
\orcid{0000-0003-1976-3182}
\affiliation{
  \institution{Cornell University}
  \city{Ithaca}
  \country{USA}
}

\author{Nate Foster}
\orcid{0000-0002-6557-684X}
\affiliation{
  \institution{Cornell University}
  \city{Ithaca}
  \country{USA}
}

\author{Tobias Kappé}
\orcid{0000-0002-6068-880X}
\affiliation{
  \institution{Leiden University}
  \city{Leiden}
  \country{The Netherlands}
}

\author{Dexter Kozen}
\orcid{0000-0002-8007-4725}
\affiliation{
  \institution{Cornell University}
  \city{Ithaca}
  \country{USA}
}

\author{Lily Saada}
\orcid{0009-0003-8590-6740}
\affiliation{
  \institution{Cornell University}
  \city{Ithaca}
  \country{USA}
}

\author{Alexandra Silva}
\orcid{0000-0001-5014-9784}
\affiliation{
  \institution{Cornell University}
  \city{Ithaca}
  \country{USA}
}

\author{Jana Wagemaker}
\orcid{0000-0002-8616-3905}
\affiliation{
  \institution{Radboud University}
  \city{Nijmegen}
  \country{The Netherlands}
}

\newcommand{\push}[1]{\ensuremath{\mathsf{push}(#1)}}
\newcommand{\pop}[1]{\ensuremath{\mathsf{pop}(#1)}}

\begin{document}
\begin{abstract}
We develop \StacKAT, a network verification language featuring loops, finite state variables, nondeterminism, and---most importantly---access to a stack with accompanying push and pop operations.
By viewing the variables and stack as the (parsed) headers and (to-be-parsed) contents of a network packet, \StacKAT can express a wide range of network behaviors including parsing, source routing, and telemetry. These behaviors are difficult or impossible to model using existing languages like \NetKAT.
We develop a decision procedure for \StacKAT program equivalence, based on finite automata.
This decision procedure provides the theoretical basis for verifying network-wide properties and is able to provide counterexamples for inequivalent programs.
Finally, we provide an axiomatization of \StacKAT equivalence and establish its completeness.
\end{abstract}
\maketitle

\section{Introduction}

\NetKAT~\cite{Anderson2014,Foster2015,Smolka2015} is a domain-specific language for specifying and verifying network behavior. Formally, a \NetKAT program is a regular expression over assignments $f \mut v$ and guards $f \test v$:\footnote{\NetKAT also has a special action $\mathsf{dup}$, which we omit here for the sake of simplicity but discuss in a later section.}
\begin{align*}
 e ::=
    0 \mid 1 \mid e_1 + e_2 \mid e_1 \cdot e_2 \mid e^* \mid f \test v \mid f \mut v
\end{align*}
Expressions like $(f \mut 3 \,+\, f \test 3 \cdot f \mut 4)^*$ represent traces similar to the regular expression $(a + b \cdot c)^*$, but with imperative actions $f \mut 3$, $f \test 3$, and $f \mut 4$ instead of letters $a$, $b$, and $c$. The action $f \mut v$ sets a variable $f$ to a constant value $v$, and the guard $f \test v$ discards the current trace if variable $f$ is not $v$. As such, \NetKAT can be viewed as an imperative programming language over a simple data model---i.e., network packets represented as finite records from header fields $f$ to values $v$.

Because \NetKAT can assign and check equality only with respect to constants, it is too simple to model general-purpose imperative programs written in languages like Python or C.
But it turns out to be an ideal model for the kind of computation that arises in high-speed network data planes.
Network devices like routers, switches, firewalls, etc., process packets using efficient piplines that classify packets using predicates formulated in terms of constant values ($f\test v$) and modify packets using actions that assign packet headers to possibly new constant values ($f \mut v$).
Moreover, \NetKAT can model not only the behavior of individual devices, but also the collective behavior of different interlinked devices, by designating a variable $\mathsf{sw}$ as the packet's current location, and representing the movement of a packet to switch $k$ as an assignment $\mathsf{sw} \mut k$.
%In this way, network-wide verification queries can be expressed as program equivalence queries. For instance, to check whether there exists a packet that can go from switch $n$ to switch $m$ under a data plane policy encoded as a \NetKAT program $e$, we can ask the equivalence query $(\mathsf{sw} \mut n) \cdot e \cdot (\mathsf{sw} \test m) \equiv 0$.

Many practical network verification questions can be formulated in terms of equivalence queries---e.g., reachability, waypointing, correct compilation, and more~\cite{Anderson2014}.
Fortunately, a \NetKAT expression $e$ can mention only finitely many values and local variables, and the values $v$ are drawn from a finite domain of numerical constants; hence, it has finite state.
As a result, semantic equivalence of \NetKAT programs is decidable, though the procedure is non-trivial~\cite{Foster2015, Moeller2024}.
%For instance, the equivalence $(f \mut 3 \,+\, f \test 3 \cdot f \mut 4)^* \equiv 1 + f \mut 3 + f \mut 4$ can be checked automatically.

However, \NetKAT's restriction to a finite state space precludes modeling a variety of behaviors that arise in practice.
First, on real-world devices, packets are not represented as finite records, but as sequences of bytes.\footnote{%
Most devices impose an upper limit on the size of a packet---i.e., the so-called Maximum Transmission Unit (MTU),
%---which is usually 1500 bytes or up to 9000 bytes with so-called ``jumbo frames'' enabled---
but mathematically, it is more natural to treat packets as having arbitrary size.
In the same way, we usually treat memory on a general-purpose computer as being unbounded, even though physical devices have a finite amount of memory.}
When a packet enters a switch it must be parsed to ``extract'' the relevant headers into local variables.
Likewise, when a packet exits a switch, the local variables must be serialized (or ``deparsed'' to use P4's term~\cite{Bosshart2014}) back into bytes.
By convention, the unparsed portion of the packet is carried along as the payload.
%Some devices parse all of the standard Internet protocols (e.g., Ethernet, VLAN, IPv4, TCP, etc.) but many do not---e.g., a simple switch might only parse the Ethernet and VLAN headers; while a router might parse Ethernet, VLAN, IPv4, and IPv6; and a firewall might parse deep into the packet including the transport or even the application layer to be able to evaluate rich access control policies.
To faithfully model packet parsing and serialization in all the various forms used on different devices, \NetKAT's finite records are insufficient.
% and a non-trivial extension is needed.

Second, certain network protocols are not finite state.
% In \emph{destination-based forwarding} schemes like IPv4, each router computes the next hop for a packet by looking up its destination address in a routing table.
In \emph{source routing} schemes~\cite{Sunshine1977}, the packet not only stores its IPv4 destination address, but encodes the entire forwarding path in a stack.
Each router pops an element off the stack and forwards the packet to the corresponding next hop.
As paths can in principle be arbitrarily long, the stack cannot be encoded using a finite record.
Dually, protocols that collect \emph{telemetry} push records onto a stack to log observations about how the packet was processed at each hop on the end-to-end path.
In similar fashion, \emph{tunneling} protocols prepend a new set of headers to the packet, encapsulating the original headers in the payload, and use the newly added headers to route through a subnetwork.
When the packet ultimately leaves the subnetwork, the new headers are removed and the original headers are restored from the payload. There are also protocols like \emph{Segment Routing} and \emph{MPLS} that combine aspects of tunneling and source routing.
%They treat the packet as a stack, such that the top of the stack determines routing behavior at the current switch, which may involve pushing, popping, and swapping one or more elements from the stack.
\NetKAT's finite records are also insufficient for modeling these protocols.
% and many other examples.

In this paper, we extend \NetKAT to \StacKAT, enriching its data model so that a packet $\langle h, s \rangle$ comprises both a finite record $h$ mapping fields to values as well as a stack $s$ representing the unparsed portion of the packet. To manipulate the record, we use \NetKAT's existing constructs. To manipulate the stack, we add $\push{v}$ and $\pop{v}$ constructs as follows:
\begin{align*}
  e ::=
    \underbrace{0 \mid 1 \mid e_1 + e_2 \mid e1 \cdot e2 \mid e^* \mid f \test v \mid f \mut v}_{\text{\NetKAT (dup-free)}} \mid
    \underbrace{\push{v} \mid \pop{v}}_{\text{New \StacKAT operations}}
\end{align*}
The new instruction $\push{v}$ pushes a value $v$ onto the stack, and $\pop{v}$ tests that the top of the stack is $v$ and removes it, dropping the packet if the top of the stack is not $v$ (similar to the guard $f \test v$).
In many examples, the stack can be thought of as a sequence of bytes that represent the rest of the packet, while the variables hold information extracted from the raw data.
Hence, with these simple extensions, \StacKAT can model parsing, source routing, telemetry, tunneling, and MPLS, as well as the $\mathsf{dup}$ instruction used to generate traces in \NetKAT's original semantics.

Various fragments of \StacKAT are obviously decidable.
For instance, equivalence of \StacKAT expressions that do not use push or pop is trivially decidable by a reduction to \NetKAT.
Also, equivalence of \StacKAT expressions that use only \textsf{push} (resp.\ \textsf{pop})---and neither local variables nor \textsf{pop} (resp.\ \textsf{push})---is also trivially decidable, by reduction to equivalence of regular expressions.
However, it is not clear whether equivalence of arbitrary \StacKAT programs is decidable, because the addition of a stack results in infinitely many states.
Moreover, equivalence of \StacKAT programs needs to take into account subtle interactions of $\push -$ and $\pop{-}$.
For instance, we will want $\push{v} \cdot \pop v \equiv 1$, as pushing a value and then immediately popping it does nothing.
Similarly, $\push{v} \cdot \pop w \equiv 0$ (for $v\neq w$), because pushing a value and then expecting a different value to be at the top of the stack will always drop the packet, just like $0$.
More subtly, $\pop{v} \cdot \push{v} \nequiv 1$ because the former rejects packets with $v$ not on top of the stack---although $\pop{v} \cdot \push{v} + 1 \equiv 1$ does hold, because the program on the left has no effect on the stack or packet, just like $1$.

%is relative to the traces they can generate, but modulo the interactions of  sequential  compositions of $\push{v}$ and  $\pop{v}$. For instance, we will want $\push{v} \cdot \pop v$ to be identity, that is equivalent to $1$. In contrast, we want $\push{v} \cdot \pop w$ to be equivalent to $0$ (for $v\neq w$) since this sequence of actions on the stack is clearly not well-typed. %With this language semantics we will be able to prove somewhat surprising program equivalences. For instance, $\push v^*\pop v^* \equiv \push v^* + \pop v^*$.

A decision procedure for \StacKAT must therefore incorporate aspects of \NetKAT equivalence (to handle local variables) and regular expression equivalence, but it must additionally handle interactions between $\push{-}$ and $\pop{-}$.
To illustrate this subtlety when deciding equivalence, consider the \StacKAT programs
$e_1 = \push{a}^* \cdot \pop{a}^* $
and
$ e_2 = \push{a}^* + \pop{a}^*$.
The program $e_1$ first pushes any number of $a$'s onto the stack, and then pops any number of $a$'s off the stack.
If there are more pushes than pops, the net effect is to push some number of $a$'s.
If there are more pops than pushes, the net effect is to pop some number of $a$'s.
This is exactly what $e_2$ does, and therefore the two are equivalent.
In contrast, $e_3 = (\push{a} \cdot \push{a})^* \cdot (\pop{a} \cdot \pop{a})^*$ is not equivalent to $e_1$, because $e_1$ can empty the stack with just $a$ on it, while $e_3$ cannot do this.
%\jules{Maybe insert a non-equivalent example here, and what the decision procedure gives as a counterexample.}
% In fact, these programs are also equivalent to $\pop{a}^* \cdot \push{a}^*$ and $(\push{a} \cdot \push{a})^* \cdot \pop{a}^*$. However, they are \emph{not} equivalent to the following programs:
% $
%   f_1 = (\pop{a}\cdot\pop{a})^* \cdot\push{a}^* $ and $f_2 = \pop{a}^* \cdot (\push{a}\cdot\push{a})^*
% $.
% For instance, $f_1$ can pop a single $a$ off the stack by doing $\pop{a}\cdot\pop{a}\cdot\push{a}$, but this only works if the initial stack contained at least two $a$'s. Similarly $f_2$ can push a single $a$ onto the stack by doing $\pop{a}\cdot\push{a}\cdot\push{a}$, but this only works if the initial stack contained at least one $a$. Therefore, $e_1$ and $e_2$ are not equivalent to $f_1$ and $f_2$. More generally, the reader may find it instructive to contemplate for which values of $n,m$ the programs $(\push{a}^n)^* \cdot (\pop{a}^m)^*$ are equivalent to $e_1$ and $e_2$, and similarly for $(\pop{a}^n)^* \cdot (\push{a}^m)^*$.

% The \StacKAT programs we have considered so far are simple expressions involving only a single numerical value $a$. In general, \StacKAT programs can be complex expressions that combine pushes and pops with multiple different values. The presence of local variables complicates the situation further. Hence, it is far from obvious how to determine whether two \StacKAT programs are equivalent.

Existing formalisms such as Visibly Pushdown Languages (VPLs)~\cite{Alur2004} handle push-pop interactions, but visibly pushdown automata execute stack actions in response to consuming symbols from an input tape.
As the equivalence problem for context-free languages is undecidable, VPLs restrict the possible stack behaviors to a decidable fragment.
This is in contrast to \StacKAT, where stack behaviors are unrestricted.
The reason that \StacKAT equivalence nevertheless remains decidable is that unlike VPLs, \StacKAT does not have a separate input tape. Rather, \StacKAT has \emph{only} a stack, whose initial contents serve as the input packet, without a separate input tape.
As a result, the notion of equivalence for \StacKAT is distinct from equivalence for VPLs.
Notably, adding an input tape to \StacKAT makes equivalence undecidable as the stack allows one to recognize context-free languages, and equivalence for context-free languages is undecidable.

Our main technical contribution is a decision procedure for \StacKAT equivalence.
In a nutshell, our technical development proceeds as follows.
We first develop a decision procedure for the ``pure'' stack fragment of \StacKAT (i.e., without local variables).
This procedure first converts the \StacKAT program to an automaton, and then canonicalizes the automaton in three steps, corresponding to three different kinds of push-pop interactions. The canonicalized automaton has the property that two \StacKAT programs are semantically equivalent if and only if the canonicalized automata accept the same language (which is decidable). We subsequently extend our decidability result to arbitrary \StacKAT programs, including those with local variables, by reduction to equivalence in the pure stack fragment. Our decision procedure not only decides equivalence, but also produces counterexamples (i.e., input-output pairs) when the given programs are inequivalent.

We furthermore axiomatize \StacKAT and prove a completeness result for the fragment without local variables.
This requires a fundamentally novel approach, because \StacKAT handles equivalences such as $\push{v} \cdot \pop{v} \equiv 1$, and $\pop{v} \cdot \push{v} + 1 \equiv 1$ that were not part of previous frameworks based on Kleene Algebra~\cite{Pous2024}.

% In fact, these axioms make equivalence undecidable in isolation \cite{Pous2024}, but our results show that we retain decidability in the presence of the additional axiom $\push{v} \cdot \pop{w} \equiv 0$ for $v \neq w$.

\noindent\paragraph{Contributions}
Overall, this paper makes the following contributions.

\begin{itemize}
\item The design of \StacKAT, a domain specific language for infinite state network verification that can model a wide range of network behaviors, such as packet parsing, source routing, telemetry, tunneling, MPLS, and more (\Cref{sec:tour})
\item A decision procedure that verifies the semantic equivalence of two \StacKAT programs, or produces a counterexample input-output pair (\Cref{sec:decpushpop,sec:decfull,sec:counterexamples})
\item A complete axiomatization of equivalence for \StacKAT without local variables (\Cref{sec:completeness})
\end{itemize}
We include an extensive discussion of related work (\Cref{sec:related}) to highlight the differences between \StacKAT and other languages and automata formalisms that go beyond regular languages and share similarities with our work.
One of the special things about \StacKAT is that it captures the local packet state, while also having access to a stack.
Together, these features make its semantics different than the systems that have been studied in the literature, to the best of our knowledge.

\ifproofatend%
For the sake of brevity, proofs appear in the appendix; we provide hyperlinks for easy navigation.
\else%
For the sake of brevity, proofs appear in the appendix of the extended version.
\fi%

% \noindent\paragraph{Outlook and Open Questions}
% %
% \StacKAT is a simple but powerful extension to \NetKAT that increases its expressiveness, enabling it to model a wide range of behaviors and protocols, including ones with infinite state. We focus in this paper on decidability of program equivalence, as it is non-trivial and provides a basis for verifying programs. Having said that, there are several interesting open questions and directions for future work.

% \begin{itemize}
%   \item While \NetKAT was introduced in 2014, in 2024~\cite{Moeller2024} a symbolic decision procedure was developed that incorporates BDD techniques to efficiently decide large equivalence queries. We do not yet know how to transfer these techniques to \StacKAT.
%   \item \NetKAT features a complete axiomatic theory. While we have an axiomatization that we conjecture is complete, we prove correctness of our decision procedure directly. New proof techniques are needed to show completeness of our axiomatization.
% \end{itemize}

% We pose these concrete challenges in \Cref{sec:futurework} to solicit help from the algorithmically and mathematically inclined reader, and we indicate the key difficulty in each case.

\begin{figure}[t]
  \renewcommand{\arraystretch}{1.2}
  \begin{tabular}{@{}c@{\hspace{8pt}}c@{\hspace{8pt}}l@{}}
      \textbf{Syntax} & \textbf{Meaning} & \textbf{Semantics $\Sem{\cdot} \subseteq \Pk \times \Pk$} \\[1mm]
      \hline \\[-0.9em]
      % $e \Coloneqq$ & & $\Sem{e} =$ \\[.25mm]
      $0$ & \emph{False} -- Drops all packets & $\emptyset$ \\[.25mm]
      $1$ & \emph{True} -- Forwards all packets & $\{ (p,p) \mid p \in \Pk \}$ \\[.25mm]
      $e_1 + e_2$ & \emph{Union} -- Send packet through $e_1$ and $e_2$ & $\Sem{e_1} \cup \Sem{e_2}$ \\[.25mm]
      $e_1 \cdot e_2$ & \emph{Sequencing} -- Send packet through $e_1$ then $e_2$ & $\{ (p,p'') \mid (p,p') \in \Sem{e_1}, (p',p'') \in \Sem{e_2} \}$ \\[.25mm]
      $e^*$ & \emph{Iteration} -- Perform $e$ any number of times & $\bigcup_{i=0}^\infty \Sem{e^i}$, where $e^0 = 1, e^{n+1} = e \cdot e^n$ \\[.25mm]
      $f \test v$ & \emph{Test equals} -- Forwards packets with $f = v$ & $\{ (\langle h, s \rangle,\langle h, s \rangle) \mid \langle h, s \rangle \in \Pk, h(f) = v \}$ \\[.25mm]
      $f \mut\, v$ & \emph{Modification} -- Sets field $f$ to $v$ & $\{ (\langle h, s \rangle, \langle h, s[f \mut v]\rangle) \mid \langle h, s \rangle \in \Pk \}$ \\[.25mm]
      $\push{v}$ & \emph{Push} -- Push $v$ onto stack & $\{ (\langle h, s \rangle, \langle h, v::s \rangle) \mid \langle h, s \rangle \in \Pk \}$ \\[.25mm]
      $\pop{v}$ & \emph{Pop} -- Test and pop top of stack if it equals $v$ & $\{ (\langle h, v::s \rangle, \langle h, s \rangle) \mid \langle h, s \rangle \in \Pk \}$ \\[3mm]
      \hline \\[-0.8em]
  \end{tabular}
  \vspace{2mm}
  \noindent
  Fields $f \in F$\quad Values $v \in \V \subseteq \N$\quad Header $h \in F \mapsto V$\quad Stack $s \in \lst(V)$\quad Packet $p \in \Pk \Coloneqq \langle h, s \rangle$

  \caption{\StacKAT syntax, meaning, and semantics.}%
  \label{fig:syn}
\end{figure}

\section{A Tour of \StacKAT}%
\label{sec:tour}

A \StacKAT program operates on packets of the form $\langle h, s \rangle \in \Pk$, comprising a finite record mapping header fields to values $h \in F \to \V$, as well as a stack $s \in \lst(\V)$. Values are drawn from a fixed, finite subset $V \subseteq \N$ of the naturals. The syntax and semantics of the \StacKAT language constructs are given in \Cref{fig:syn}. The formal semantics $\Sem{e} \subseteq \Pk \times \Pk$ of a \StacKAT expression is a relation between input and output packets, such that an input packet produces zero or more output packets. A pair of \StacKAT programs $e$ and $f$ are equivalent (written $e \equiv f$) if they denote the same relation on input and output packets ($\Sem e = \Sem f$). We have, for instance
\[
\push{v} \cdot \pop{v} \equiv 1, \text{\qquad but \qquad} \pop{v} \cdot \push{v} \nequiv 1,
\]
as $(\langle h, s \rangle, \langle h, s \rangle) \in \Sem{1}$ for all stacks $s$, but $(\langle h, s \rangle, \langle h, s \rangle) \in \Sem{\pop{v} \cdot \push{v}}$ only for stacks that start with $v$ on top.
% \jana{I would leave this out (for space), and also the next example}
% As a first example of a \StacKAT program consider the following expression:
% \[
% \push{v} \cdot \pop{v}
% \]
% We can calculate the semantics of $\push{v} \cdot \pop{v}$ as follows:
% \begin{align*}
% &\steps{\langle h, s \rangle}{\langle h'', s'' \rangle}{\push{v} \cdot \pop{v}}\\
%  \iff &\steps{\langle h, s \rangle}{\langle h', s' \rangle}{\push{v}} \quad\text{ and }\quad \steps{\langle h', s' \rangle}{\langle h'', s'' \rangle}{ \pop{v}}\\
%   \iff &  h = h', s' = v :: s \quad\text{ and } h' = h'', s' = v :: s'' \\
%   \Rightarrow & h = h'' \text{ and } s = s''
% \end{align*}
% Hence, $\Sem{\push{v} \cdot \pop{v}}$ is the identity relation $\{(\langle h, s \rangle, \langle h, s \rangle) \mid \langle h, s \rangle\in\Pk  \}$ which is precisely the semantics of $1$, the program that forwards all packets. In other words, we just saw a first example of a valid \StacKAT equivalence: $\push{v} \cdot \pop{v} \equiv 1$.
Consider now the \StacKAT equivalence
\[
 \push{v}^* \cdot \pop{v}^* \ \equiv\ \push{v}^* + \pop{v}^*
\]
At first sight, these programs appear different---the first one sequentially executes an arbitrary number of $\push v$ followed by an arbitrary number of $\pop v$ whereas the second one branches and then either does sequences of only $\push v$ or only $\pop v$. However, in the \StacKAT semantics these are equivalent because doing some number of $\push{v}$'s followed by some other number of $\pop{v}$'s will, on net after canceling out $\push{v}\cdot \pop{v} \equiv 1$, result in either only $\push{v}$'s or only $\pop{v}$'s, depending on whether the number of pushes or pops was larger.

\paragraph{Exercise} Which of the following \StacKAT programs are equivalent?
\begin{align*}
  e_1 &\triangleq \pop{v}^* \cdot \push{v}^* &&&
  e_2 &\triangleq (\push{v} + \pop{v})^* \\
  e_3 &\triangleq \push{v}^* \cdot (\pop{v} \cdot \pop{v})^* &&&
  e_4 &\triangleq (\pop{v}\cdot\pop{v})^* \cdot\push{v}^* \\
  e_5 &\triangleq (\push{v} \cdot \push{v})^* \cdot \pop{v}^* &&&
  e_6 &\triangleq \pop{v}^* \cdot (\push{v}\cdot\push{v})^*
\end{align*}

\textbf{The main goal of this paper is to design a decision procedure that enables us to automatically check program equivalence of \StacKAT programs,} not just for simple examples involving a single value $v$ as above, but for arbitrary \StacKAT programs involving multiple different values as well as local variables. Furthermore, if two programs are not equivalent, we would like to find a counterexample input-output pair. As we will see, designing such a procedure in the presence of the stack-related equivalences induced by the semantics of \StacKAT requires care and is not a simple reduction to \NetKAT or other well-studied automata frameworks.

% $
%   f_1 = (\pop{a}\cdot\pop{a})^* \cdot\push{a}^* $ and $f_2 = \pop{a}^* \cdot (\push{a}\cdot\push{a})^*
% $.
% For instance, $f_1$ can pop a single $a$ off the stack by doing $\pop{a}\cdot\pop{a}\cdot\push{a}$, but this only works if the initial stack contained at least two $a$'s. Similarly $f_2$ can push a single $a$ onto the stack by doing $\pop{a}\cdot\push{a}\cdot\push{a}$, but this only works if the initial stack contained at least one $a$. Therefore, $e_1$ and $e_2$ are not equivalent to $f_1$ and $f_2$. More generally, the reader may find it instructive to contemplate for which values of $n,m$ the programs $(\push{a}^n)^* \cdot (\pop{a}^m)^*$ are equivalent to $e_1$ and $e_2$, and similarly for $(\pop{a}^n)^* \cdot (\push{a}^m)^*$.

\subsection{Syntactic Sugar}

\StacKAT is a minimal core calculus, designed for exploring theoretical questions. Nevertheless, in presenting examples, it will be useful to use various derived constructs, which can be defined as syntactic sugar.
General Boolean expressions can be supported using $+$ for disjunction and $\cdot$ for conjunction, and negation via De Morgan's laws. For instance, inequality of two fields:
\begin{align*}
  f \testNE g &\ \triangleq\  \neg \sum_{v \in V} f \test v \cdot g \test v\ \equiv\ \prod_{v \in V} (f \testNE v + g \testNE v)
  \text{\qquad where \qquad} f \testNE v \triangleq \sum_{v' \in V \setminus \{ v \}} f \test v'
\end{align*}
We can use these Boolean expressions inside conditionals and loops~\cite{Kozen1996}:
\begin{align*}
  \textbf{if } f \test v \textbf{ then } e_1 \textbf{ else } e_2 \triangleq\  f \test v \cdot e_1 + f \testNE v \cdot e_2 &&
  \textbf{while } f \test v \textbf{ do } e \triangleq\  (f \test v \cdot e)^* \cdot f \testNE v
\end{align*}
We can also define syntactic sugar for assigning one field to another and for comparing fields:
\begin{align*}
  f \mut g \triangleq\  \sum_{v \in V} f \test v \cdot g \mut v &&
  f \test g \triangleq\  \sum_{v \in V} f \test v \cdot g \test v
\end{align*}
Similarly, we can define syntactic sugar for pushing a field onto the stack, and popping the top off the stack and storing the result in a field:
\begin{align*}
  \push{f} \triangleq\  \sum_{v \in V} f \test v \cdot \push{v} &&
  \pop{f} \triangleq\  \sum_{v \in V} \pop{v} \cdot f \mut v
\end{align*}
% Finally, although the formal semantics models the stack in terms of a list of bytes, in examples, we will sometimes abuse notation and use local variables and stack operations that operate at other widths---e.g., Ethernet addresses are 6 bytes, but we will simply write $\push{ether\mathord{-}dst}$ rather than $\push{ether\mathord{-}dst_6} \cdot \dots \cdot \push{ether\mathord{-}dst_1}$.\footnote{Note that to keep things simple we will ignore endian issues in this paper.} More generally, if non-byte-aligned acess were required, we could formulate the stack as simply a list of bits without changing the essential character of the language. For the sake of simplicity in this paper, we will assume that working at byte granularity is sufficient.

Note that the definitions of many of these derived constructs rely on the finiteness of the value domain $V$. Of course, enumerating the finite (but presumably very large) domain of values would be totally impractical in an implementation. However, this is not an unreachable approach for a mathematical treatment of the language and its core properties, as with \NetKAT.

% \begin{figure}[t]
%   \begin{mathpar}
%     \inferrule*[right=True]{\ }{\steps{p}{p}{1}} \and
%     \inferrule*[right=Sequence]{\steps{p}{p'}{e_1} \quad \steps{p'}{p''}{e_2}}{\steps{p}{p''}{e_1 \cdot e_2}} \and
%     \inferrule*[right=Star]{\steps{p}{p'}{1 + e \cdot e^*}}{\steps{p}{p'}{e^*}} \and
%     \inferrule*[right=Union-L]{\steps{p}{p'}{e_1}}{\steps{p}{p'}{e_1 + e_2}} \and
%     \inferrule*[right=Union-R]{\steps{p}{p'}{e_2}}{\steps{p}{p'}{e_1 + e_2}} \\
%     \inferrule*[right=Modify]{h' = h[f \mut v]}{\steps{\langle h, s \rangle}{\langle h', s \rangle}{f \mut v}} \and
%     \inferrule*[right=Test]{h(f) = v}{\steps{\langle h, s \rangle}{\langle h, s \rangle}{f = v}} \and
%     \inferrule*[right=Push]{s' = v :: s}{\steps{\langle h, s \rangle}{\langle h, s' \rangle}{\push{v}}} \and
%     \inferrule*[right=Pop]{s = v :: s'}{\steps{\langle h, s \rangle}{\langle h, s' \rangle}{\pop{v}}}
%   \end{mathpar}
%   \caption{\StacKAT formal semantics. Inference rules define for every \StacKAT expression $e$ a relation $\Sem{e} \subseteq \Pk \times \Pk$.}
%   \label{fig:sem}
% \end{figure}

\subsection{What Can Be Expressed in \StacKAT}%
\label{sec:applications}

\StacKAT can be used to encode the behavior of a wide range of networking applications, including complex behaviors that would be difficult or impossible to model in \NetKAT. We first discuss how network behavior is modeled in both \NetKAT and \StacKAT, and then discuss \StacKAT's increased expressiveness compared to \NetKAT.

\paragraph{Modeling the Behavior of a Single Switch}
The behavior of a switch is modeled as a program $e$ that takes a packet as input and produces zero or more packets as output.
The header fields of the packet are modeled as local variables.
One special field $\mathsf{sw}$, called the switch field, is used to track the current switch that the packet resides on.
To move the packet to a different switch, we assign $\mathsf{sw} \mut n$ for the target switch $n$.
Typically, the behavior of a switch program $e$ is to inspect some of the packet's header fields (using $\text{if } f \test v$ \text{ then \ldots\ else \ldots}), and then either drop the packet (using $0$), or send it to a neighboring switch (using $\mathsf{sw} \mut n$), possibly after modifying some of the packet's header fields (using $f \mut v$). A switch can also output more than one packet (using $+$), such as when a switch forwards a packet to multiple neighbors.

\paragraph{Modeling the Behavior of a Network of Switches}
Let $e_1, e_2, \dots, e_n$ be the programs that model the behavior of the switches in the network.
The behavior of the entire network is then given by the program $e \triangleq (\mathsf{sw}\test 1 \cdot e_1 + \mathsf{sw}\test 2 \cdot e_2 + \cdots + \mathsf{sw}\test n \cdot e_n)^*$. This program first tests which switch the packet is currently on, and then behaves as the program for that switch. The Kleene star is used to encode the fact that a packet can traverse multiple switches.

\paragraph{Decidability of Equivalence}
Whether two \NetKAT programs are equivalent is decidable, because variables only take on finitely many different values.
Equivalence of programs can be used to express a wide range of properties about network behavior, e.g.,
\begin{itemize}
  \item Reachability: is there any packet that can go from switch $i$ to switch $j$? ($\mathsf{sw}\test i \cdot e \cdot \mathsf{sw}\test j \equiv 0$)
  \item Slice isolation: can the network be partitioned into two virtual networks? ($e \equiv e_1 + e_2$)
\end{itemize}
A key goal of \StacKAT is to remain decidable even as we add more expressive constructs.

\paragraph{Increased Expressiveness of \StacKAT}
\StacKAT's increased expressiveness comes from the ability to model not only the header fields of the packet, but also the payload or remainder of the packet, which can be manipulated using $\push{v}$ and $\pop{v}$. This allows to model the following features, which are impossible to express in \NetKAT:
\begin{itemize}
  \item \textbf{Packet parsing and unparsing:} whereas \NetKAT assumes that the header fields of the packet are already parsed, the pop instruction allows \StacKAT to model parsing the start of the payload of the packet into header fields, and the push instruction to unparse the header fields back into the payload.
  \\\emph{Example.} Parsing headers: $\pop{f_1} \cdot \pop{f_2}$.
  Serializing headers: $\push{f_1} \cdot \push{f_2}$.
  \item \textbf{Source routing:} in source routing, the packet contains a list of routing instructions that determine the desired path of the packet. \StacKAT can model source routing by storing the routing instructions in the packet and then popping them off the stack at each hop.
  \\\emph{Example.} Go to new switch based on top of stack: $\pop{\text{sw}}$.
  \item \textbf{Tunneling:} tunneling protocols are used to transport packets through sub-networks that use different protocols to carry packets. \StacKAT can model tunneling by pushing the new headers onto the stack when the packet enters the sub-network, and stripping them off when exiting the network.
  \\\emph{Example.} Enter tunnel: $\push{f_1} \cdot \push{f_2} \cdot f_1 \mut v_1 \cdot f_2 \mut v_2$. Exit tunnel: $\pop{f_1} \cdot \pop{f_2}$.
  \item \textbf{Segment Routing and MPLS:} In segment routing and MPLS, each switch removes a routing instruction as in source routing, but may then also add new instructions onto the packet that cause it to follow the right segment. \StacKAT can model this hierarchical behavior found in both Segment Routing and MPLS by pushing and popping values.
  \\\emph{Example.} Do stack transformation: $\pop{v_1} \cdot \push{v_2} \cdot \push{v_3} + \pop{w_1} \cdot \push{w_2}$.
  \item \textbf{String matching:} \StacKAT can model regular matching over routing instructions by nesting pop instructions in a regular expression.
  \\\emph{Example.} Match string: $\pop{a} \cdot (\pop{b} + \pop{c})^*$.
  \item \textbf{Telemetry:} \StacKAT can model simple forms of telemetry by pushing tracked values.
  \\\emph{Example.} Push telemetry value: $\push{t}$.
\end{itemize}

\paragraph{\NetKAT's $\mathit{dup}$ Operator} \StacKAT can model \NetKAT's $\mathsf{dup}$ operator, which appends the current packet to a trace for observing internal network behavior. By pushing all fields onto the stack ($\mathsf{dup} \triangleq \prod_{f \in F} \push{f}$), \StacKAT simulates \NetKAT's trace semantics---e.g., $f \mut 1 \cdot \mathsf{dup} \cdot f \mut 2 \nequiv f \mut 2$ even though $f \mut 1 \cdot f \mut 2 \equiv f \mut 2$. This encoding preserves \NetKAT equivalence: two \NetKAT programs are equivalent if and only if their \StacKAT translations are equivalent.

\paragraph{Combinations of Features}
The combination of features above is possible (e.g., parsing, tunneling, MPLS, etc.), provided that only a single stack is used. Use of two multiple stacks makes the equivalence problem undecidable~\cite{Hopcroft2006}. In fact, this remains true even if one of the stacks is constrained to be read-only or write-only (e.g., MPLS + telemetry on separate stacks), as the problem is then the same as equivalence for context-free grammars, which is known to be undecidable~\cite{Hopcroft2006}.

\section{Decidability of Equivalence}%
\label{sec:decpushpop}

This section presents our procedure for deciding equivalence of \StacKAT programs:
\begin{align*}
  \text{Given } e_1 \text{ and } e_2, \text{ decide whether } e_1 \equiv e_2 \text{ or } e_1 \nequiv e_2.
\end{align*}
We first present a decision procedure for pure push-pop \StacKAT programs without local variables, i.e., without the operations $f \test v$ and $f \mut v$. Our decision procedure for full \StacKAT programs will invoke this procedure as a subroutine. Programs in the push-pop fragment can be viewed as regular expressions over the alphabet
\(
\Sigma = \{\push{v} \mid v \in V\} \cup \{\pop{v} \mid v \in V\}
\).
In this view, each word in the language accepted by the regular expression represents a trace of push and pop operations.

We wish to decide equivalence between push-pop programs via this interpretation, and the connection between regular languages and automata. The regular language $L(e)$ of a push-pop program is not the same as $\Sem{e}$, but a connection can be obtained through a series of canonicalizations that take into account the following non-trivial equivalences induced by the semantics (\Cref{fig:syn}).

First, pushing a value $v$ and then popping $v$ is the identity:
\begin{align*}
  \push{v} \cdot \pop{v} \equiv 1 \phantom{\text{\qquad if } v \neq w} \tag{push-pop}
\end{align*}

Second, pushing $v$ and then popping a different value $v \neq w$ fails:
\begin{align*}
  \push{v} \cdot \pop{w} \equiv 0 \text{\qquad if } v \neq w \tag{filter}
\end{align*}

Third, popping $v$ and then pushing $v$ back is \emph{almost} the identity:
\begin{align*}
  \pop{v} \cdot \push{v} + 1 \equiv 1 \qquad\qquad \tag{pop-push}
\end{align*}

The intuition for the third property is that $\pop{v} \cdot \push{v}$ does not alter the stack, provided that the top of the stack is $v$, and therefore its behaviors are included in those of $1$.

Our approach to deciding equivalence proceeds as follows:
\begin{enumerate}
  \item We first convert the \StacKAT programs $e$ and $f$ to finite automata that accept $L(e)$ and $L(f)$.
  \item We then perform a series of three canonicalization steps corresponding to the three axioms above to obtain canonical automata that accept the canonicalized languages $\overline{L(e)}$ and $\overline{L(f)}$.
  \item Finally, we decide whether $\overline{L(e)} = \overline{L(f)}$ are language equivalent.
\end{enumerate}

We then show that language equivalence of the canonicalized languages implies semantic equivalence of the original \StacKAT programs:
\begin{align*}
  \Sem{e}=\Sem{f} \iff \overline{L(e)}=\overline{L(f)}
\end{align*}

We proceed to describe the canonicalization steps in detail, and then show how to extend the decision procedure to \StacKAT programs with local variables.

\subsection{Push-Pop Canonicalization}

\newcommand{\pushpop}[1]{\ensuremath{\mathsf{pushpop}(#1)}}

Let us assume that we have a \StacKAT program $e$ that contains only push and pop operations, and that we have a language $L(e)$ that represents the traces of push and pop operations that $e$ can perform, as well as an NFA that accepts $L(e)$. This NFA can be constructed from regular expression $e$ using standard constructions, such as Antimirov derivatives~\cite{Antimirov1996}.

The first step of canonicalization reduces matching $\push{v} \cdot \pop{v}$ pairs. We define the operation $\pushpop{L}$ on a language $L$, which closes $L$ under the reduction $\push{v} \cdot \pop{v} \to \epsilon$. In order to make it possible to perform the canonicalization on automata as well, $\pushpop{L}$ adds reduced traces while keeping original traces in $L$. This ensures that $\pushpop{L}$ is a regular language whenever $L$ is a regular language. Let us consider an example.

For the singleton language $L = \{ \push{1} \cdot \push{2} \cdot \pop{2} \cdot \pop{1} \cdot \push{3} \cdot \pop{3} \}$,
\begin{align*}
  \pushpop{L} = \{ &\push{1} \cdot \push{2} \cdot \pop{2} \cdot \pop{1} \cdot \push{3} \cdot \pop{3},\\
  &\push{1} \cdot \pop{1} \cdot \push{3} \cdot \pop{3},\ \
  \push{1} \cdot \push{2} \cdot \pop{2} \cdot \pop{1},\\
  &\push{3} \cdot \pop{3},\ \
  \push{1} \cdot \pop{1},\ \
  \epsilon \}
\end{align*}
That is, reduced strings are added for all possible reduction orders.
The $\pushpop{L}$ canonicalization on sets of strings is defined in \Cref{fig:pushpoplang}.

\begin{figure}
  \begin{mathpar}
    \inferrule{x \in L}{x \in \pushpop{L}} \and
    \inferrule{x \cdot \push{v} \cdot \pop{v} \cdot y \in \pushpop{L}}{x \cdot y \in \pushpop{L}}
  \end{mathpar}
  \caption{Push-pop canonicalization of a language $L$}%
  \label{fig:pushpoplang}
\end{figure}

\paragraph{Push-Pop Canonicalization for the Trace Automaton}

To perform the canonicalization on the automaton, we add shortcut $\epsilon$ edges for every $\push{v}\cdot \epsilon^* \cdot \pop{v}$ path in the automaton. In the implementation, to avoid $\epsilon^*$ paths, we simultaneously perform $\epsilon$ closure.

The rules for adding epsilon edges are shown in \Cref{fig:pushpopaut}.
First, we add an $\epsilon$ self-loop to each state.
Then, for each $\push{v}-\epsilon-\pop{v}$ path of length 3, we add an $\epsilon$ shortcut edge.
Finally, for each $\epsilon-\epsilon$ path of length 2, we add an $\epsilon$ shortcut edge.
If we perform this closure on an automaton for $L(e)$, the language accepted by the resulting canonicalized automaton is $\pushpop{L(e)}$.

\begin{theoremEnd}[default]{lemma}
  For any NFA over the alphabet $\Sigma = \{\push{v}, \pop{v} \mid v \in V\}$ representing language $L$, the language accepted by the NFA after performing the push-pop canonicalization is $\pushpop{L}$.
\end{theoremEnd}
\begin{proofEnd}
  Let $A$ be an NFA over the alphabet $\Sigma = \{\push{v}, \pop{v} \mid v \in V\}$ representing language $L$. We construct an NFA $\pushpop{A}$ according to the rules of \Cref{fig:pushpopaut}.
  That is, we have a sequence of automata $A_0 = A$, $A_1$, $A_2$, \dots, where $A_{i+1}$ is obtained from $A_i$ by adding $\epsilon$ edges according to an application of one of the rules in \Cref{fig:pushpopaut}. Because at most $n^2$ $\epsilon$-edges can be added to the automaton (one edge for each pair of states), the process stabilizes. The resulting automaton $A_n = \pushpop{A}$ by definition, and application of the rules results in no further changes.

  We then show that $\pushpop{L(A)} = L(\pushpop{A})$ by proving inclusions $\subseteq$ and $\supseteq$.
  \begin{itemize}
    \item $\subseteq$:
    We need to show $\pushpop{L(A)} \subseteq L(\pushpop{A})$, i.e., all strings in the push-pop language are accepted by the push-pop automaton.
    We show this by induction on the rules in \Cref{fig:pushpoplang}.
    For the first rule, we need to show that if $x \in L(A)$ then $x \in L(\pushpop{A})$, which follows from the definition of $\pushpop{A}$, which only adds edges to the automaton.
    For the second rule, we need to show that if $x \cdot \push{v} \cdot \pop{v} \cdot y \in L(\pushpop{A})$, then $x \cdot y \in L(\pushpop{A})$.
    By the induction hypothesis, there is a path in $\pushpop{A}$ that generates $x \cdot \push{v} \cdot \pop{v} \cdot y$.
    Between every letter in the string, the path may have followed any number of $\epsilon$ edges.
    In particular, the path may have followed $\epsilon$ edges between $\push{v}$ and $\pop{v}$.
    By the last two closure rules in \Cref{fig:pushpopaut}, the path may be replaced by a single $\epsilon$ edge between the $\push{v}$ and $\pop{v}$.
    Therefore, the first closure rule will have added an $\epsilon$ edges that skips the intervening $\push{v}$ and $\pop{v}$.
    The result is a path that generates $x \cdot y$, so $x \cdot y \in L(\pushpop{A})$.
    \item $\supseteq$:
    We need to show $L(\pushpop{A}) \subseteq \pushpop{L(A)}$, i.e., all strings accepted by the push-pop automaton are in the push-pop language.
    We show this by induction on the sequence $A_0 = A, A_1, \ldots, \pushpop{A}$:
    we show that for all $i$, $L(A_i) \subseteq \pushpop{L(A)}$.
    For $i=0$, $L(A) \subseteq \pushpop{L(A)}$ follows from the definition of $\pushpop{L(A)}$, which only adds strings to the language.
    For $i>0$, we assume that $L(A_{i-1}) \subseteq \pushpop{L(A)}$, and we show that $L(A_i) \subseteq \pushpop{L(A)}$. Consider the closure rule that was applied to go from $A_{i-1}$ to $A_i$.
    If the rule was the second or third rule, then $L(A_i) = L(A_{i-1})$, and the inclusion follows from the induction hypothesis.
    If the rule was the first rule, then the language does change, as follows:
    for strings of the form $\vec{x} \cdot \push{v} \cdot \pop{v} \cdot \vec{y}$ generated by going from the start state to $q_1$, then taking the $\push{v}$ edge to $q_2$, then taking the $\epsilon$ edge to $q_3$, and taking the $\pop{v}$ edge to $q_4$ and then continuing to a final state, the string $\vec{x} \cdot \vec{y}$ is added to the language accepted by the automaton. By assumption, the string $\vec{x} \cdot \push{v} \cdot \pop{v} \cdot \vec{y}$ is in $\pushpop{L(A)}$.
    By the definition of $\pushpop{L(A)}$, the string $\vec{x} \cdot \vec{y}$ is also in $\pushpop{L(A)}$.
    Therefore, $L(A_i) \subseteq \pushpop{L(A)}$, in particular for $A_n = \pushpop{A}$, we have $L(\pushpop{A}) \subseteq \pushpop{L(A)}$.
  \end{itemize}
\end{proofEnd}

\begin{figure}
  \begin{tikzpicture}[shorten >=1pt, node distance=2cm, on grid, auto]

    % States for the push-eps-pop diagram
    \node[state] (q1) at (1.2,0) {$q_1$};
    \node[state] (q2) [right=of q1] {$q_2$};
    \node[state] (q3) [right=of q2] {$q_3$};
    \node[state] (q4) [right=of q3] {$q_4$};

    % Arrows for the push-eps-pop diagram
    \path[->]
      (q1) edge node {$\push{v}$} (q2)
      (q2) edge node {$\epsilon$} (q3)
      (q3) edge node {$\pop{v}$} (q4)
      (q1) edge[bend right, dashed, swap] node {$\epsilon$} (q4);

    % States for the eps-eps diagram
    \node[state] (r1) at (8.4,0) {$q_1$};
    \node[state] (r2) [right=of r1] {$q_2$};
    \node[state] (r3) [right=of r2] {$q_3$};

    % Arrows for the eps-eps diagram
    \path[->]
      (r1) edge node {$\epsilon$} (r2)
      (r2) edge node {$\epsilon$} (r3)
      (r1) edge[bend right, dashed, swap] node {$\epsilon$} (r3);

    % Self-loop diagram
    \node[state] (s1) at (0,0) {$q_1$};

    % Arrow for the self-loop
    \path[->]
      (s1) edge[loop below, in=225, out=315, looseness=8, dashed] node {$\epsilon$} (s1);

  \end{tikzpicture}
  \caption{Push-pop closure rules. Given the presence of the solid edges, we add the dotted edge. For each $\push{v}-\epsilon-\pop{v}$ path, we add an $\epsilon$ shortcut edge, and we simultaneously take the reflexive transitive closure of the $\epsilon$ edges.}%
  \label{fig:pushpopaut}
\end{figure}
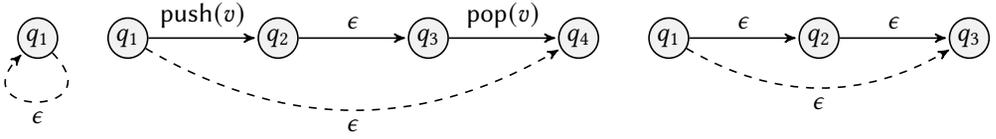

\subsection{Filtering Out Invalid Traces}

\newcommand{\Lpoppush}{L_{\text{pop}^*\text{push}^*}}

The next step is to filter out invalid traces that contain $\push{v}\cdot\pop{w}$ with $v \neq w$. It is important that the push-pop canonicalization from the preceding section is performed first, as this can reveal additional $\push{v}\cdot\pop{w}$ combinations by reducing an intervening substring to $\epsilon$, such as $\push{v}\cdot\push{a}\cdot\pop{a}\cdot\pop{w} \to \push{v}\cdot\pop{w}$.

Having performed the push-pop closure first, we also know that we have reduced any $\push{v} \cdot \pop{w}$ with $v=w$, which makes the original trace with $\push{v} \cdot \pop{w}$ redundant. Therefore, we filter out \emph{all} traces where a $\push{v}$ occurs before a $\pop{w}$, regardless of whether $v \neq w$ or $v = w$. In other words, we only keep the traces of the form
\begin{align*}
  \Lpoppush \triangleq (\pop{v_1} + \cdots + \pop{v_n})^* \cdot (\push{v_1} + \cdots + \push{v_n})^*
\end{align*}
where $V = \{ v_1, \dots, v_n \}$ is the set of values that occur.
We therefore define
\begin{align*}
  \mathsf{filter}(L) \triangleq L \cap \Lpoppush
\end{align*}
The intersection of two regular languages is regular, and we can compute an automaton for $\mathsf{filter}(L)$ from the automaton for $L$.
%However, in the next section we will see that we do not need to perform the filtering on the automaton level separately, as the pop-push closure in the next section will already filter out invalid traces naturally.

\paragraph{Example}
We now show a simple example of the procedure on $e = \push{3}^*\cdot \pop{3}^*$ and $e' = \push{3}^* + \pop{3}^*$ to check that they are equivalent.

In the first step, an automaton is constructed for $e = \push{3}^*\cdot \pop{3}^*$ (left), then we take the push-pop closure (middle), and finally we take the intersection with pop-push (right).

\medskip
\begin{minipage}{0.33\textwidth}
 \begin{tikzpicture}[state/.style={draw,ellipse}, shorten >=1pt, node distance=2cm, on grid, auto,every node/.style={font=\footnotesize},inner sep=.05cm]
    % States for the \push{3}^*;\pop{3}^*
    \node[state] (q1) at (1,0) {$e$};
    \node[state] (q2) [right=of q1] {$e_0$};
    % \node[state] (q3) at (2,-1.5){$e_1$};

    % Arrows for the \push{3}^*;\pop{3}^*
    \path[->]
      (q1) edge[loop below] node {$\push{3}$} ()
      (q1) edge node {$\pop{3}$} (q2)
      (q2) edge[loop below] node {$\pop{3}$} ();
      % (q1) edge node {$\push{3}$} (q2)
      % (q1) edge node [swap] {$\pop{3}$} (q3)
      % (q2) edge node  {$\pop 3$} (q3)
      % (q2) edge[loop  right] node {$\push{3}$} ()
      % (q3) edge[loop  right] node {$\pop{3}$} ();
  \end{tikzpicture}
\end{minipage}
\begin{minipage}{0.33\textwidth}
  \begin{tikzpicture}[state/.style={draw,ellipse}, shorten >=1pt, node distance=2cm, on grid, auto,every node/.style={font=\footnotesize},inner sep=.05cm]
    % States for the \push{3}^*;\pop{3}^*
    \node[state] (q1) at (1,0) {$e$};
    \node[state] (q2) [right=of q1] {$e_0$};
    % \node[state] (q3) at (2,-1.5){$e_1$};

    % Arrows for the \push{3}^*;\pop{3}^*
    \path[->]
      (q1) edge[loop below] node {$\push{3}$} ()
      (q1) edge node {$\pop{3}$} (q2)
      (q2) edge[loop below] node {$\pop{3}$} ()
      (q1) edge[red, bend right, dashed, swap] node {$\epsilon$} (q2);
      % (q1) edge node {$\push{3}$} (q2)
      % (q1) edge node [swap] {$\pop{3}$} (q3)
      % (q2) edge node  {$\pop 3$} (q3)
      % (q2) edge[loop  right] node {$\push{3}$} ()
      % (q3) edge[loop  right] node {$\pop{3}$} ();
  \end{tikzpicture}
%  \begin{tikzpicture}[state/.style={draw,ellipse}, shorten >=1pt, node distance=2cm, on grid, auto,every node/.style={font=\footnotesize},inner sep=.05cm]
%     % States for the \push{3}^*;\pop{3}^*
%     \node[state] (q1) at (1,0) {$e$};
%     \node[state] (q2) [right=of q1] {$e_0$};
%     \node[state] (q3) at (2,-1.5){$e_1$};

%     % Arrows for the \push{3}^*;\pop{3}^*
%     \path[->]
%       (q1) edge node {$\push{3}$} (q2)
%       (q1) edge node [swap] {$\pop{3}$} (q3)
%        (q1) edge[red,bend left]  node  {$\epsilon$} (q3)
%        (q2) edge[red,bend right] node  [swap] {$\epsilon$} (q3)
%       (q2) edge node {$\pop 3$} (q3)
%       (q2) edge[loop  right] node {$\push{3}$} ()
%       (q3) edge[loop  right] node {$\pop{3}$} ();
%   \end{tikzpicture}
 \end{minipage}
 \begin{minipage}{0.33\textwidth}
 \begin{tikzpicture}[state/.style={draw,ellipse}, shorten >=1pt, node distance=2cm, on grid, auto,every node/.style={font=\footnotesize},inner sep=.05cm]
    % States for the \push{3}^*;\pop{3}^*
    \node[state] (q1) at (1,0) {$e$};
    \node[state] (q2) [right=of q1] {$e_1$};
    \node[state] (q3) at (2,-0.5){$e_2$};

    % Arrows for the \push{3}^*;\pop{3}^*
    \path[->]
      (q1) edge node {$\pop{3}$} (q2)
      (q1) edge node [swap] {$\push{3}$} (q3)
          (q2) edge[loop  right] node {$\pop{3}$} ()
      (q3) edge[loop  right] node {$\push{3}$} ();
  \end{tikzpicture}
 \end{minipage}
\medskip

 The automaton on the right is also the automaton one would build for $e' = \push{3}^* + \pop{3}^*$ (in the very first step, as in this case taking push-pop closure, intersect with pop-push does not change the automaton). Hence, the automata for the expressions  $e$ and $e'$ are bisimilar, and as we shall see, it follows that $\Sem{\push{3}^*\cdot \pop{3}^*} = \Sem{\push{3}^* + \pop{3}^*}$.

More generally, we can consider the program $(\push 3^n)^* (\pop 3^m)^*$, for any $n,m\in \mathbb N$, and show that it is in fact equivalent to $(\push 3^k)^* + (\pop 3^k)^*$ where $k=\gcd(n,m)$.

\subsection{Pop-Push Canonicalization}

\newcommand{\poppush}[1]{\ensuremath{\mathsf{poppush}(#1)}}

We have now reduced traces where $\push{v}$ immediately occurs before $\pop{v}$ and limited ourselves to traces where all $\pop{v}$ operations happen before $\push{w}$ operations. The reader may wonder whether we are now done, as all remaining traces are valid and in reduced form. However, we have not yet taken into account the fact that $\pop{v}\cdot \push{v}$ is almost the identity, which results in an interaction between different traces. Consider the following trace sets:
\begin{align*}
  L_1 \triangleq \{ \epsilon,\ \pop{v} \cdot \push{v} \} \qquad\qquad L_2 \triangleq \{ \epsilon \}
\end{align*}
corresponding to the programs $e_1 = 1 + \pop{v} \cdot \push{v}$ and $e_2 = 1$.
These trace sets are in reduced form and not equal, but they are semantically equivalent.

At first sight, it may seem attractive to canonicalize the trace sets by removing the redundant trace $\pop{v} \cdot \push{v}$. However, removing this trace is only valid because the $\epsilon$ trace is present. In general, we cannot remove traces without considering the full context---i.e., possible interactions between different traces. We conjecture that the resulting operation would be exceedingly difficult to define and compute, particularly on the automaton level.

Instead, we define a canonicalization operation $\poppush{L}$ that \emph{adds} traces by inserting $\pop{v}\cdot \push{v}$ in the middle of a trace in all possible ways. The canonicalization is defined in \Cref{fig:poppushlang}.

\begin{figure}
  \begin{mathpar}
    \inferrule{x \in L}{x \in \poppush{L}} \and
    \inferrule{x \cdot y \in \poppush{L} \and \text{$x$ all pop and $y$ all push}}{x \cdot \pop{v} \cdot \push{v} \cdot y \in \poppush{L}}
  \end{mathpar}
  \caption{Pop-push canonicalization of a language $L$}%
  \label{fig:poppushlang}
\end{figure}

This transformation is enough to fully canonicalize the trace language, so that semantic equivalence coincides with language equivalence, as we shall see in the next section. Unfortunately $\poppush{L}$ is not regular, so we cannot compute an automaton for it. In fact, the language is not even visibly pushdown\footnote{This language is not visibly pushdown, because it is not known for each alphabet letter whether it is a call (push), return (pop) or a transition action. For instance a letter $\pop{v}$ in a word $w\in\poppush{L}$ can be both a transition action (in case the letter already appeared in $w$ before the closure) or a return action (in case the letter is a result of the $\mathsf{poppush}$-closure).}~\cite{Alur2004}, so decidability of language equivalence is not a priori guaranteed.
We address this difficulty by shifting perspectives and switching to a different representation.

\paragraph{Zipping the Trace Language}
\newcommand{\zip}[1]{\ensuremath{\mathsf{zip}(#1)}}
\newcommand{\done}{\mathsf{done}}

To address the non-regularity of the language, we switch to a different representation of the trace language. We define the zipped trace language $\mathsf{zip}(L)$, which represents a trace from the middle (i.e., the border between \textsf{pop}s and \textsf{push}es) outwards. For example,
\begin{align*}
  \tau = \pop{1} \cdot \pop{2} \cdot \push{3} \cdot \push{4} &&&
  \zip{\tau} = (\pop{2}, \push{3}) \cdot (\pop{1}, \push{4})
\end{align*}
When there are more pushes than pops or vice versa, we use $\done$ as a placeholder. For example,
\begin{align*}
  \tau' = \pop{1} \cdot \push{2} \cdot \push{3} &&&
  \zip{\tau'} = (\pop{1}, \push{2}) \cdot (\done, \push{3})
\end{align*}
We define the \textsf{zip} operation formally in \Cref{fig:ziplang}.
For strings that are not of the form $\mathsf{pop}^*\mathsf{push}^*$, the zipping operation is undefined. For example, the trace $\tau'' = \pop{1} \cdot \push{2} \cdot \pop{3}$ is not zippable, i.e. $\zip{\tau''} = \bot$.
We extend \textsf{zip} to sets of traces by applying it to each trace in the set:
\begin{align*}
  \zip{L} \triangleq \{\zip{x}\mid x\in L \text{ and }\zip{x} \text{ defined }\}
\end{align*}
This not only zips the language, but also filters out invalid traces, as the zipping operation is only defined for traces of the form $\mathsf{pop}^*\mathsf{push}^*$.

The key advantage of \textsf{zip} is that whereas $\poppush{L}$ is not regular, $\zip{\poppush{L}}$ \emph{is}.
In fact, the pop-push canonicalization can easily be performed on the zipped trace language, and the result is manifestly regular.
As zipping already filters out invalid traces, and is a bijection on the $\Lpoppush$ fragment, we can check the equivalence of the zipped languages instead of the original languages:

\begin{theoremEnd}[default]{lemma}
  $\zip{\filter(L)} = \zip{L}$
\end{theoremEnd}
\begin{proofEnd}
  This is immediate as $\zip{x}$ for a trace $x$ is non-empty only when $x\in \Lpoppush$.
\end{proofEnd}

\begin{theoremEnd}[default]{lemma}
  $ \zip{L_1} = \zip{L_2} \iff L_1 = L_2$ for $L_1,L_2 \subseteq \Lpoppush$
\end{theoremEnd}
\begin{proofEnd}
  Immediate from the definition of zipping, which is injective on $\Lpoppush$.
\end{proofEnd}

The following lemma characterizes the zipped trace language of the pop-push canonicalization of a language $L$.
Let $A \triangleq \{ (\pop{v},\push{v}) \mid v \in V \}$.

\begin{theoremEnd}[default]{lemma}
  $\zip{\poppush{L}} = A^* \cdot \zip{L}$
\end{theoremEnd}
\begin{proofEnd}
  We show that the claim is true for a singleton language $L = \{ x \}$, which implies the claim for any language $L$ because both sides respect union.
  Firstly, consider whether $x$ is of the form $\pop{\vec{v}}\cdot\push{\vec{w}}$.
  If not, then both sides are empty, and the claim holds.
  If $x$ is of that form, we have:
  \begin{align*}
    \poppush{L} =\ &\{
      \pop{\vec{v}}\cdot
      \pop{\vec{u}^R} \cdot \push{\vec{u}}\cdot
      \push{\vec{w}} \mid \vec{u} \in \Sigma^* \}\\
      =\ &\ \pop{\vec{v}}\cdot\{
        \pop{\vec{u}^R} \cdot \push{\vec{u}} \mid \vec{u} \in \Sigma^* \}
        \cdot
        \push{\vec{w}}
  \end{align*}
  where $\vec{u}^R$ is the reverse of $\vec{u}$.
  Therefore
  \begin{align*}
    \zip{\poppush{L}} =
      \underbrace{\zip{\{\pop{\vec{u}^R} \cdot \push{\vec{u}} \mid \vec{u} \in \Sigma^* \}}}_{A^*}\ \cdot\ \
      \zip{\underbrace{\pop{\vec{v}}\cdot\push{\vec{w}}}_{x}}
  \end{align*}
  So indeed, $\zip{\poppush{L}} = A^* \cdot \zip{L}$.
\end{proofEnd}
% \begin{proof}
% For the left-to-right direction, consider a word $w\in\zip{\poppush{L}}$. Hence, there exists $x\in\poppush{L}$ such that $\zip{x}=w$. From the definition of $\zip{x}$, we can derive that $x=\pop{v_1}\cdots\pop{v_n}\cdot\push{w_1}\cdots\push{w_m}$ for some $v$'s and $w$'s (note that $n$ and $m$ can be $0$, making $x$ the empty word). We prove by induction on the structure of $\poppush{L}$ that for any word $u\in \poppush{L}$ of the shape $\pop{v_1}\cdots\pop{v_n}\cdot\push{w_1}\cdots\push{w_m}$, we have $\zip{u}\in A^*\cdot\zip{L}$. This would conclude the proof. For the base case we consider $u\in L$, in which case we are done immediately. Otherwise we have $u=u_1\cdot\pop{v}\cdot\push{v}\cdot u_n$ for $u_1$ all pop and $u_n$ all push, where we obtain from the induction hypothesis that $\zip{u_1\cdot u_n}\in A^*\cdot\zip{L}$. As $\zip{u}=(\pop{v},\push{v})\cdot\zip{u_1\cdot u_n}$, we can conclude.

% For the right-to-left direction, we take $w\in A^*\cdot\zip{L}$, which implies that $w=x_1x_2$ such that $x_1\in A^*$ and $x_2=\zip{y}$ for $y\in L$. From the definition of $\zip{y}$, we can derive that $y=\pop{v_1}\cdots\pop{v_n}\cdot\push{w_1}\cdots\push{w_m}$ for some $v$'s and $w$'s. Furthermore, we know that $x_1=(\pop{u_1},\push{u_1})\cdots(\pop{u_k},\push{u_k})$ for some $u$'s. As $y\in L$, we conclude that $t=\pop{v_1}\cdots\pop{v_n}\cdot\pop{u_k}\cdots\pop{u_1}\cdot\push{u_1}\cdots\push{u_k}\cdot\push{w_1}\cdots\push{w_m}\in\poppush{L}$. It is immediate that $\zip{t}=x_1x_2$, which concludes the proof.
% \end{proof}

\begin{figure}
  \begin{align*}
    \zip{\epsilon} &\triangleq \epsilon \\
    \zip{ x \cdot \pop{v} \cdot \push{w} \cdot y} &\triangleq (\pop{v},\push{w}) \cdot \zip{x\cdot y}\\
    \zip{x \cdot \pop{v}} &\triangleq (\pop{v}, \done) \cdot \zip{x} \\
    \zip{\push{v} \cdot y} &\triangleq (\done, \push{v}) \cdot \zip{y}\\
    \zip{a} &\triangleq  \text{undefined otherwise}
  \end{align*}
  \caption{Zipping the trace language; in these equations,, $x \in \pop{V}^*$, $y \in \push{V}^*$}%
  \label{fig:ziplang}
\end{figure}

\paragraph{Zipping the Trace Automaton}

\begin{figure}
  \begin{mathpar}
    \inferrule
      {q_1 \xrightarrow{\pop{v}} q'_1 \and q_2 \xrightarrow{\push{w}} q'_2}
      {(q'_1,q_2) \xrightarrow{(\pop{v},\push{w})} (q_1,q'_2)} \quad
    \inferrule{q_1 \mathsf{\ initial}}
      {(q_1,q_2) \xrightarrow{\epsilon} (\done,q_2)} \quad
    \inferrule{q_2 \mathsf{\ final}}
      {(q_1,q_2) \xrightarrow{\epsilon} (q_1,\done)} \quad
      \inferrule{\ }{(q,q) \mathsf{\ initial}}
      \\
    \inferrule{q_1 \xrightarrow{\pop{v}} q'_1}
      {(q'_1,\done) \xrightarrow{(\pop{v},\done)} (q_1,\done)} \quad
    \inferrule{q_2 \xrightarrow{\push{v}} q'_2}
      {(\done,q_2) \xrightarrow{(\done,\push{v})} (\done,q'_2)} \quad
    \inferrule{\ }{(\done,\done) \mathsf{\ final}}
      \\
    \inferrule{q_2 \xrightarrow{\epsilon} q'_2}
      {(q_1,q_2) \xrightarrow{\epsilon} (q_1,q'_2)} \quad
    \inferrule{q_2 \xrightarrow{\epsilon} q'_2}
      {(\done,q_2) \xrightarrow{\epsilon} (\done,q'_2)} \quad
    \inferrule{q_1 \xrightarrow{\epsilon} q'_1}
      {(q'_1,q_2) \xrightarrow{\epsilon} (q_1,q_2)} \quad
    \inferrule{q_1 \xrightarrow{\epsilon} q'_1}
      {(q'_1,\done) \xrightarrow{\epsilon} (q_1,\done)}
  \end{mathpar}
  \caption{Zipping the trace automaton}%
  \label{fig:zipaut}
\end{figure}

In order to make use of \textsf{zip} for the decision procedure, we need to define a zipping operation on the automaton level.
Given an automaton for a language $L$, the zipped automaton for $\zip{L}$ should traverse the original automaton ``from the middle outwards'', taking a backward pop step and a forward push step simultaneously.
Given a NFA for language $L$, we construct an NFA for $\zip{L}$ as follows.

\begin{description}
  \item[States] $(q_1,q_2)$ where $q_1$ and $q_2$ are either states of the original automaton, or an extra state $\done$.
  \item[Initial states] $(q,q)$ where $q$ is any state of the original automaton.
  \item[Final states] $(\done,\done)$ only.
  \item[Transitions] A transition in the zipped automaton corresponds to a simultaneous backward pop and forward push transition in the original automaton. We also add $\epsilon$ transitions to $\done$ for the initial and final states. If the first or second component of the state is $\done$, we can step in the zipped automaton with a single step in the original automaton. The transitions are defined in \Cref{fig:zipaut}.
\end{description}

\begin{theoremEnd}[default]{lemma}\label{lemma:zipregular}
  For any NFA over the alphabet $\Sigma = \{\push{v}, \pop{v} \mid v \in V\}$ representing language $L$, the language accepted by the zipped NFA is $\zip{L}$.
\end{theoremEnd}
\begin{proofEnd}
  We first show that any string in $\zip{L}$ is accepted by the zipped automaton.
  If $z \in \zip{L}$, then there exists $x \in \filter(L)$ such that $z = \zip{x}$.
  Therefore, there exists a path $q_i$ of some length $n$ in the original automaton from a start state to a final state, corresponding to the letters of $x$.
  We construct a path in the zipped automaton as follows.
  First, remove $\epsilon$-transitions from the path, giving a path $q'_i$ of length $n' \neq n$.
  Let $m$ be the ``middle'' state $q'_m$ on the path when the string transitions from the last pop to the first push.
  Now we shall construct a path in the zipped automaton.
  Let $Q'_0 = (q'_m,q'_m)$ as the first state of the zipped path (which is a start state in the zipped automaton).
  More generally, let $Q'_i = (q'_{m-i},q'_{m+i})$ be the states on the zipped path, taking $q'_k = \done$ for $k<0$ or $k\geq n$ until $Q'_N = (\done, \done)$.
  The path $Q'$ is not yet a path in the zipped automaton, because we have removed $\epsilon$-transitions. However, for every $i$ there is a path with additional $\epsilon^*$-transition from $Q'_i$ to $Q'_{i+1}$ in the zipped automaton, so we can expand the path $Q'$ to a longer path $Q'$ that takes these $\epsilon$-transitions and accepts the string $z$.
  Therefore, any string in $\zip{L}$ is accepted by the zipped automaton.

  For the other direction is straightforward. Take a path in the zipped automaton. Transforming this path into a path in the original automaton: for start state $(q,q)$ in the zipped automaton, we start in state $q$ in the original automaton, and then each step in the zipped automaton extends this path in either both directions or in a single direction.
\end{proofEnd}

With the zipped trace automaton in \Cref{fig:zipaut}, we can can now compute the pop-push closure on the automaton level by precomposing the automaton with the $A^*$ language.

\begin{theoremEnd}[default]{lemma}
  For any NFA over the alphabet $\Sigma = \{\push{v}, \pop{v} \mid v \in V\}$ representing language $L$, the language accepted by the NFA after performing the push-pop canonicalization, zipping, and precomposing with $A^*$ is $\zip{\poppush{\mathsf{filter}(\pushpop{L})}}$.
\end{theoremEnd}
\begin{proofEnd}
  Composition of earlier lemmas.
\end{proofEnd}

\paragraph{Example}
In the preceding example, the automata for $e = \push{3}^* \cdot \pop{3}^*$ and $e' = \push{3}^* + \pop{3}^*$ already became equivalent after the first two steps of canonicalization, and the final step of zipping and pop-push canonicalization was not necessary. However, in general, the final step is necessary to decide equivalence of push-pop programs as the following example illustrates.
Consider the program $e'' = \pop{3}^* \cdot \push{3}^*$ (note that, compared to $e$, the order of \textsf{push} and \textsf{pop} differs here). An automaton for this program is shown on the left below.

The push-pop closure of this automaton is the same as the original automaton, and the intersection with the pop-push automaton is also the same. The resulting automaton is not bisimilar to the final automata of $e$ and $e'$, yet we know that $e''$ is semantically equivalent to $e$ and $e'$. The pop-push closure is needed to decide this equivalence. We therefore construct the zipped automata for $e''$ and $e$. These are shown below; in the middle is the zipped automaton for $e''$, and on the right the zipped automaton for $e$. We have omitted the $\epsilon$ transitions and merged equivalent states for clarity.

\medskip
\begin{minipage}[t]{0.3\textwidth}
  \centering
 \begin{tikzpicture}[state/.style={draw,ellipse}, shorten >=1pt, node distance=2cm, on grid, auto,every node/.style={font=\footnotesize},inner sep=.05cm]
    % States for the \pop{3}^*;\push{3}^*
    \node[state] (q1) at (1,0) {$e$};
    \node[state] (q2) [right=of q1] {$e_0$};
    \node[state] (q3) at (2,-1.5){$e_1$};
    \phantom{ \node[state] (q4) at (2,-1.3) {$e_2$}; }

    % Arrows for the \pop{3}^*;\push{3}^*
    \path[->]
      (q1) edge node {$\pop{3}$} (q2)
      (q1) edge node [near start, swap] {$\push{3}$} (q3)
      (q2) edge node[near start]  {$\push 3$} (q3)
      (q2) edge[loop  right] node[above=4pt] {$\pop{3}$} ()
      (q3) edge[loop  right] node {$\push{3}$} ();
  \end{tikzpicture}
\end{minipage}
\begin{minipage}[b]{0.34\textwidth}
  \centering
 \begin{tikzpicture}[state/.style={draw,ellipse}, shorten >=1pt, node distance=2cm, on grid, auto,every node/.style={font=\footnotesize},inner sep=.05cm,node distance=1.2cm]
    % States for the \pop{3}^*;\push{3}^*
    \node[state] (q2) {};
    \node[state] (q3) [below left=of q2] {};
    \node[state] (q4) [below right=of q2] {};

    % Arrows for the \pop{3}^*;\push{3}^*
    \path[->]
      % (q1) edge node {$(\pop{3},\push{3})$} (q2)
      % (q1) edge node [swap] {$\push{3}$} (q3)
      (q2) edge node [swap, near start]  {$(\done,\push 3)$} (q3)
      (q2) edge[loop  above] node {$(\pop{3},\push{3})$} ()
      (q3) edge[loop  below] node {$(\done,\push{3})\quad$} ()
      (q2) edge node [near start] {$(\pop{3},\done)$} (q4)
      (q4) edge[loop  below] node {\quad$(\pop{3},\done)$} ();
  \end{tikzpicture}
\end{minipage}
\begin{minipage}[b]{0.34\textwidth}
  \centering
 \begin{tikzpicture}[state/.style={draw,ellipse}, shorten >=1pt, node distance=2cm, on grid, auto,every node/.style={font=\footnotesize},inner sep=.05cm,node distance=1.2cm]
    % States for the \pop{3}^*;\push{3}^*
    % \node[state] (q1) at (1,0) {$$};
    \node[state] (q2) {};
    \node[state] (q3) [below left=of q2] {};
    \node[state] (q4) [below right=of q2] {};

    % Arrows for the \pop{3}^*;\push{3}^*
    \path[->]
      % (q1) edge node {$(\pop{3},\push{3})$} (q2)
      % (q1) edge node [swap] {$\push{3}$} (q3)
      (q2) edge node [swap, near start] {$(\done,\push 3)$} (q3)
      % (q2) edge[loop  right] node {$(\pop{3},\push{3})$} ()
      (q3) edge[loop  below, swap] node[swap] {$(\done,\push{3})\quad$} ()
      (q2) edge node [near start] {$(\pop{3},\done)$} (q4)
      (q4) edge[loop  below] node {\quad$(\pop{3},\done)$} ();
  \end{tikzpicture}
\end{minipage}

 In order to then perform the pop-push closure on these automata, we simply prepend the language $(\pop{3},\push{3})^*$ to the automata. This has no effect on the language of the automaton on the left, as it already starts with $(\pop{3},\push{3})^*$. The automaton on the right, however, is extended with a $(\pop{3},\push{3})$ loop on the leftmost state. This crucial last step makes the resulting automata language equivalent, and next we will see that this allows us to conclude that $\Sem{e''} = \Sem{e}$.

\subsection{Deciding Push-Pop Programs}
For a \StacKAT program $e$ in the push-pop fragment, let $L(e)$ be the language obtained by viewing $e$ as a regular expression over $\Sigma = \{\push{v},\pop{v} \mid v \in v\}$. Further, define
\begin{align*}
  \overline{L} &\triangleq \poppush{\mathsf{filter}(\pushpop{L})}
\end{align*}
Note that the notation $\overline{L}$ does \emph{not} denote the complement of $L$, as it sometimes does in other work, but rather this specific operation.
We have the following theorem:
\begin{theorem}\label{thm:semcanon}
  For two push-pop programs $e$ and $f$, we have $\Sem{e} = \Sem{f} \iff \overline{L(e)} = \overline{L(f)}$.
\end{theorem}

\newcommand{\emp}{\mathsf{emp}}

To prove this theorem, we define a semantic interpretation of a language $L$ as follows:
\begin{align*}
  [L]=\bigcup_{x\in L}\Sem{x}
\end{align*}
i.e., the semantic interpretation of a language is the union of the semantic interpretations of its strings, viewed as \StacKAT programs (a string can be viewed as a program in the obvious way, where the empty string is the program $1$).
The following lemma connects the semantic interpretation of a push-pop program with the language of its traces.
\begin{theoremEnd}[default]{lemma}\label{lem:semtraces}
  For a push-pop program $e$, we have $\Sem{e} = [L(e)]$.
\end{theoremEnd}
\begin{proofEnd}
  By induction on $e$.
\end{proofEnd}
i.e., the semantics of a push-pop program agrees with the semantics of all its push-pop traces.

Second, the canonicalization rules are semantically valid w.r.t. $\sem{-}$:
\begin{theoremEnd}[default]{lemma}\label{lem:semcanonsame}
  $\sem{L} = \sem{\overline{L}}$
\end{theoremEnd}
\begin{proofEnd}
  Easy consequence of the definitions.
\end{proofEnd}
We define an alternative semantic interpretation of a language $L$ as follows:
\begin{align*}
  \sem{L}' = \{(\langle \emp,s \rangle, \langle \emp,s' \rangle)\mid \pop{s}\push{s'}\in L\}
\end{align*}
where $\emp$ is the empty packet header, and $\pop{s}$ on an entire stack $s=s_1::\ldots :: s_n$ (as opposed to a single value) is defined as $\pop{s_1}\cdots\pop{s_n}$, and similarly, $\push{s}$ is defined as $\push{s_n}\cdots\push{s_1}$.
On languages of the form $\overline{L}$, these two semantics agree:
\begin{theoremEnd}[default]{lemma}\label{lem:semprime}
  $\sem{\overline{L}} = \sem{\overline{L}}'$
\end{theoremEnd}
\begin{proofEnd}
  The inclusion $\supseteq$ is trivial, and the inclusion $\subseteq$ follows because $[L]'$ differs from $[L]$ only in that $[L]$ allows stepping with stacks that may contain additional values below the strings in $L$, whereas $[L]'$ requires the stack to exactly match the strings in $L$. However, $\poppush{L}$ adds additional strings to $L$ with $\pop{v}\push{v}$ pairs in the middle, which precisely match the possible additional values in the stack.
\end{proofEnd}
The inclusion $\supseteq$ is trivial, and the inclusion $\subseteq$ follows because $[L]'$ differs from $[L]$ only in that $[L]$ allows stepping with stacks that may contain additional values below the strings in $L$, whereas $[L]'$ requires the stack to exactly match the strings in $L$. However, $\poppush{L}$ adds additional strings to $L$ with $\pop{v}\push{v}$ pairs in the middle, which precisely match the possible additional values in the stack.

We also have the following (trivial) lemma:
\begin{theoremEnd}[default]{lemma}\label{lem:semprimetriv}
  If strings in $L_1,L_2$ are of the form $\pop{s}\push{s'}$, then
    $[L_1]' = [L_2]' \iff L_1 = L_2$
\end{theoremEnd}
\begin{proofEnd}
  Easy consequence of the definitions.
\end{proofEnd}

We can now return to \Cref{thm:semcanon}.
\begin{proof}[Proof (of \Cref{thm:semcanon})]
By \Cref{lem:semtraces}, it suffices to prove
$
  [L(e)] = [L(f)] \iff \overline{L(e)} = \overline{L(f)}
$.
By \Cref{lem:semcanonsame}, it suffices to prove
$
  [\overline{L(e)}] = [\overline{L(f)}] \iff \overline{L(e)} = \overline{L(f)}
$.
By \Cref{lem:semprime}, it suffices to prove
$
  [\overline{L(e)}]' = [\overline{L(f)}]' \iff \overline{L(e)} = \overline{L(f)}
$.
We obtain the desired result by \Cref{lem:semprimetriv}.
\end{proof}

We conclude that equivalence of \StacKAT programs in the push-pop fragment is decidable:

\begin{theorem}
  To decide whether $\Sem{e} = \Sem{f}$, convert $e$ and $f$ to NFAs, perform the push-pop closure, zipping, and pop-push closure, obtaining automata for $\zip{\overline{L(e)}}$ and $\zip{\overline{L(f)}}$, and check bisimilarity of these resulting automata.
\end{theorem}

 \subsection{Computational Complexity}
We establish the computational complexity of the decision problem for the push-pop fragment of \StacKAT, as well as a matching hardness result. Checking equivalence of \StacKAT programs with only push operations and no pop operations (or vice versa) is identical to checking language equivalence of these programs viewed as regular expressions, as no interaction between the two types of operations is possible if only one type of operation is present. Checking language equivalence of regular expressions is known to be PSPACE-complete. We therefore conclude:

\begin{theorem}
  The equivalence problem for the push-pop fragment of \StacKAT is PSPACE-hard.
\end{theorem}

The problem remains in PSPACE even in the presence of \emph{both} push and pop operation, as our decision procedure builds a polynomially sized automaton from the input, and then performs a NFA equivalence check, which can be done in PSPACE;\@ thus, the problem is PSPACE-complete.

\begin{theorem}
  Our decision procedure for the push-pop fragment of \StacKAT is in PSPACE\@.
\end{theorem}

\section{Decidability of Equivalence for Full \StacKAT}%
\label{sec:decfull}

\newcommand{\head}{\alpha}
\newcommand{\Head}{\mathsf{Head}}
\newcommand{\stack}{\sigma}
\newcommand{\traces}[3]{\ensuremath{\mathsf{traces}_{#1}^{#2}(#3)}}
\newcommand{\canontraces}[3]{\ensuremath{\overline{\mathsf{traces}}_{#1}^{#2}(#3)}}

We now extend the decision procedure to the full \StacKAT language, including tests and modifications, by reducing the problem to the push-pop fragment. For any given input and output packet headers $\head_1$ and $\head_2$, we define the trace language $\traces{\head_1}{\head_2}{e}$ of a \StacKAT program $e$ that can be observed when the input packet header is $\head_1$ and the output packet header is $\head_2$. We will present a method for constructing an NFA for $\traces{\head_1}{\head_2}{e}$. We then show that the equivalence problem $\Sem{e} = \Sem{f}$ can be reduced to the equivalence problem for the trace languages $\traces{\head_1}{\head_2}{e}$ and $\traces{\head_1}{\head_2}{f}$ for all $\head_1$ and $\head_2$, and that the latter problem can be solved using the decision procedure for the push-pop fragment.

To start, we define the trace language of a \StacKAT program with tests and modifications.
We have a trace language for any given input and output packet header values.
For example, consider
\begin{align*}
  e \triangleq (f \test 1 \cdot \push{1} + f \mut 2 \cdot \push{2})^*
\end{align*}
The trace languages $\traces{\head_1}{\head_2}{e}$ fall into several cases. The trace language is:
\begin{itemize}
  \item $\push{1}^* \cdot \push{2}^*$ for input packet header $f=1$, and output packet header $f = 2$.
  \item $\push{1}^*$ for input packet header $f=1$, and output packet header $f = 1$.
  \item $\push{2}^*$ for input packet header $f \neq 1$, and output packet header $f = 2$.
  \item $\emptyset$ for input packet header $f \neq 1$, and output packet header $f \neq 2$.
\end{itemize}

Intuitively, $\traces{\head_1}{\head_2}{e}$ captures an over-approximation of stack manipulations found in $e$ that can take place while $e$ takes input packet header $\head_1$ to output packet header $\head_2$.
The trace language is formally defined in \Cref{fig:tracelang}, by recursion on the structure of the program $e$.

\subsection{The Trace Automaton}

The trace language $\traces{\head_1}{\head_2}{e}$ for given input and output headers is a regular language, which we prove by showing that it can be represented by a finite automaton. Furthermore, the explicit construction of this automaton shows how to implement a possible decision procedure.
We represent states in the automaton as pairs $\langle \head, e \rangle$, where $\head$ is the current packet header and $e$ is the current expression. \Cref{fig:traceaut} defines the transitions of the trace automaton.
The automaton takes $\epsilon$ transitions for tests and modifications, and to decompose the union and star operations. For a sequence $e_1 \cdot e_2$, the automaton takes a transition on $e_1$ until $e_1$ is finished, and then continues with $e_2$.
For push and pop, the automaton takes a transition labeled with the corresponding operation.

To construct the NFA for $\traces{\head_1}{\head_2}{e}$, we start with the initial state $\langle \head_1, e \rangle$ and take transitions according to the rules in \Cref{fig:traceaut}. The final states are those of the form $\langle \head_2, 1 \rangle$.

We also need to show that the resulting automaton is always finite. The $\head$ component of $\langle \head, e \rangle$ ranges over a finite set, but it may not be clear that the $e$ component only reaches a finite number of distinct expressions, because the star rule $e^* \to e \cdot e^*$ can enlarge the expression. However, the automaton steps only to expressions of the form $e_1 \cdot e_2 \cdots e_k$ where each $e_i$ is strictly smaller than $e_{i+1}$, and $e_k$ smaller than $e$, where ``smaller'' means that the expression is a subexpression, possibly with certain occurrences of $f \test v$, $f \mut v$, $\push{v}$, $\pop{v}$ replaced by $1$. Therefore, the number of distinct reachable expressions is finite.

\begin{theoremEnd}[default]{lemma}
The language accepted by the NFA for $\traces{\head_1}{\head_2}{e}$ is indeed $\traces{\head_1}{\head_2}{e}$.
\end{theoremEnd}
\begin{proofEnd}
  Analogous to the same proof for Antimirov derivatives.
\end{proofEnd}

\begin{figure}
  \begin{align*}
    \traces{\head}{\head'}{0} &\triangleq \emptyset &&&
    \traces{\head}{\head'}{1} &\triangleq \begin{cases}
      \{ \epsilon \} & \text{if $\head = \head'$} \\
      \emptyset & \text{otherwise}
    \end{cases}\\
    \traces{\head}{\head'}{f \test v} &\triangleq \begin{cases}
      \{ \epsilon \} & \text{if $\head = \head'$ and $\head_f = v$} \\
      \emptyset & \text{otherwise}
    \end{cases} &&&
    \traces{\head}{\head'}{f \mut v} &\triangleq \begin{cases}
      \{ \epsilon \} & \text{if $\head' = \head[f \mut v]$} \\
      \emptyset & \text{otherwise}
    \end{cases}\\
    \traces{\head}{\head'}{\push{v}} &\triangleq \begin{cases}
      \{ \push{v} \} & \text{if $\head = \head'$} \\
      \emptyset & \text{otherwise}
    \end{cases} &&&
    \traces{\head}{\head'}{\pop{v}} &\triangleq \begin{cases}
      \{ \pop{v} \} & \text{if $\head = \head'$} \\
      \emptyset & \text{otherwise}
    \end{cases}
  \end{align*}
  \begin{align*}
    \traces{\head}{\head'}{e_1 + e_2} &\triangleq \traces{\head}{\head'}{e_1} \cup \traces{\head}{\head'}{e_2} \\
    \traces{\head}{\head'}{e_1 \cdot e_2} &\triangleq \{ s_1 \cdot s_2 \mid s_1 \in \traces{\head}{\beta}{e_1},\ s_2 \in \traces{\beta}{\head'}{e_2}, \beta \in F \to V \} \\
    \traces{\head_1}{\head_2}{e^*} &\triangleq \bigcup_{n \in \mathbb{N}} \traces{\head_1}{\head_2}{e^n}, \qquad \text{where $e^0 \triangleq 1$ and $e^{n+1} \triangleq e^n \cdot e$}
  \end{align*}
  \caption{The trace language of a \StacKAT program.}%
  \label{fig:tracelang}
\end{figure}

\begin{figure}
  \begin{mathpar}
    % test can step to 1 if the packet matches
    \inferrule{\head_f = v}{\langle \head, f \test v \rangle \xrightarrow{ \epsilon } \langle \head, 1 \rangle} \and
    % mutation can step to the mutated packet
    \inferrule{}{\langle \head, f \mut v \rangle \xrightarrow{ \epsilon } \langle \head[f \mut v], 1 \rangle} \and
    \inferrule{}{\langle \head, \push{v} \rangle \xrightarrow{ \push{v} } \langle \head, 1 \rangle} \and
    \inferrule{}{\langle \head, \pop{v} \rangle \xrightarrow{ \pop{v} } \langle \head, 1 \rangle} \and
    % union case steps to components with epsilon transitions
    \inferrule{}{\langle \head, e_1 + e_2 \rangle \xrightarrow{ \epsilon } \langle \head, e_1 \rangle} \and
    \inferrule{}{\langle \head, e_1 + e_2 \rangle \xrightarrow{ \epsilon } \langle \head, e_2 \rangle} \and
    % sequence case can step in front of the first component
    \inferrule{\langle \head, e_1 \rangle \xrightarrow{ s } \langle \head', e_1' \rangle}{\langle \head, e_1 \cdot e_2 \rangle \xrightarrow{ s } \langle \head', e_1' \cdot e_2 \rangle} \and
    % sequence case can step in the second component if the first is finished
    \inferrule{}{\langle \head, 1 \cdot e \rangle \xrightarrow{ \epsilon } \langle \head, e \rangle} \and
    % star can step to 1
    \inferrule{}{\langle \head, e^* \rangle \xrightarrow{ \epsilon } \langle \head, 1 \rangle} \and
    % star can step to the sequence of e and e*
    \inferrule{}{\langle \head, e^* \rangle \xrightarrow{ \epsilon } \langle \head, e \cdot e^* \rangle}
  \end{mathpar}
  \caption{The transitions of the trace automaton of a \StacKAT program.}%
  \label{fig:traceaut}
\end{figure}

\subsection{Decidability of Equivalence}

We have shown how to construct an automaton for the trace language and how to perform the canonicalization steps on the automaton level. We obtain the canonical automaton by putting together steps for the push-pop closure, zipping, and pop-push closure:
\begin{align*}
  \canontraces{\head_1}{\head_2}{e} \triangleq \poppush{\zip{\pushpop{\traces{\head_1}{\head_2}{e}}}}
\end{align*}

We can use the canonical automaton to build a decision procedure for equivalence:
\begin{theorem}\label{thm:main}
  $\canontraces{\head_1}{\head_2}{e_1} = \canontraces{\head_1}{\head_2}{e_2}$ for all $\head_1, \head_2$ if and only if $\Sem{e_1} = \Sem{e_2}$
\end{theorem}

The condition on the traces can be checked by language equivalence of the automata, which is decidable. Second, the packet headers $\head_1,\head_2$ range over a finite set. Therefore, we have a decision procedure for equivalence of \StacKAT programs.
The proof of \Cref{thm:main} is entirely analogous to \Cref{thm:semcanon}, except that \Cref{lem:semtraces} is replaced by the following lemma:

\begin{theoremEnd}[default]{lemma}
  For a \StacKAT program $e$, we have $\Sem{e} = \bigcup_{\head_1,\head_2} [\traces{\head_1}{\head_2}{e}]_{(\head_1,\head_2)}$
\end{theoremEnd}
\begin{proofEnd}
  We prove two inclusions separately. For one direction we do induction on $e$, and for the other direction we do induction on the semantics.
\end{proofEnd}
Here $[\traces{\head_1}{\head_2}{e}]_{(\head_1,\head_2)}$ is defined as adding the packet headers to the pure stack traces:
\begin{align*}
  [\traces{\head_1}{\head_2}{e}]_{(\head_1,\head_2)} \triangleq \{ ((\head_1, s), (\head_2, s')) \mid (\langle \emp,s \rangle, \langle \emp,s' \rangle) \in [\traces{\head_1}{\head_2}{e}] \}
\end{align*}
On the right-hand side of this definition, note that $\traces{\head_1}{\head_2}{e}$ contains words over letters of the form $\push{v}$ and $\pop{v}$ exclusively, and so the pairs in $[\traces{\head_1}{\head_2}{e}]$ are always of the form $(\langle h, s\rangle, \langle h, s'\rangle)$ --- i.e., with the output packet the same as the input packet.
The rest of the proof proceeds the same way as for \Cref{thm:semcanon}, but using $[L]_{(\head_1,\head_2)}$ instead of $[L]$.

\subsection{Computational Complexity}
We establish the computational complexity of the decision procedure for full \StacKAT, as well as a matching hardness result.
For the hardness result, we reduce equivalence of regular expressions with squaring (which is EXPSPACE-complete~\cite{Meyer1972}) to equivalence of push-only \StacKAT programs with variables. Squaring makes the problem harder because it allows expressions to be encoded more compactly as $e^2$ instead of $e \cdot e$, which would otherwise lead to an expression that is twice as large. \StacKAT variables can compactly encode squaring:
 $e^2$ can be modeled by $x \gets 0 \cdot ((x=0 + x=1) \cdot e \cdot (x=0 \cdot x\gets1 + x=1\cdot x\gets 2))^* \cdot x \neq 0 \cdot x \neq 1$, picking $x$ fresh for each square that occurs, initializing $x \gets 0$ at the beginning. In more common parlance, $x\gets 0 ; \text{while} (x=0 \lor x=1) \{ e ; \text{ if } x=0 \text{ then } x \gets 1 \text{ else if } x=1 \text{ then } x\gets 2 \}$ runs $e$ twice without duplicating $e$ textually. This increases the size of the expression by only a constant factor, and leads to the following result:

\begin{theorem}
  The equivalence problem for full \StacKAT is EXPSPACE-hard.
\end{theorem}

Our decision procedure matches this bound, as the trace automaton of a \StacKAT program is at most exponential in the size of the program, and the subsequent steps are PSPACE:\@
\begin{theorem}
  Our decision procedure for full \StacKAT is in EXPSPACE\@.
\end{theorem}

Note that the encoding of $e^2$ requires one additional fresh variable per square; the equivalence problem with a fixed set of variables is still PSPACE (as this bounds the size of the trace automaton).

\section{Counterexample Extraction and Implementation}%
\label{sec:counterexamples}

The decision procedure for equivalence of \StacKAT programs can be used to extract counterexamples. Given two programs $e$ and $f$ that are not equivalent, the procedure will find a packet header pair $\head_1, \head_2$ such that $\canontraces{\head_1}{\head_2}{e} \neq \canontraces{\head_1}{\head_2}{f}$. The packet headers $\head_1$ and $\head_2$ can be used to construct a counterexample input-output pair for which the programs $e$ and $f$ differ.

Because $\canontraces{\head_1}{\head_2}{e} \neq \canontraces{\head_1}{\head_2}{f}$, the bisimilarity check will find a difference in the automata. The difference can be used to extract a word that is accepted by one automaton but not the other.
Because the final automaton has been zipped, this counterexample word will be of the form $\push{v_1} \cdots \push{v_n} \pop{w_m} \cdots \pop{w_1}$ (after unzipping the word again), where the $v_i,w_i$ are values of the stack symbols. Viewing $\vec{v} = (v_1,\ldots,v_n)$ and $\vec{w} = (w_1,\ldots,w_m)$ as stacks, we can construct a counterexample input-output pair as follows:
\begin{align*}
  \text{Input: } \langle \head_1, \vec{v} \rangle \quad \text{Output: } \langle \head_2, \vec{w} \rangle \text{\quad where \quad}
  (\langle \head_1, \vec{v} \rangle, \langle \head_2, \vec{w} \rangle) \in (\Sem{e} \setminus \Sem{f}) \cup (\Sem{f} \setminus \Sem{e})
\end{align*}
This input-output pair is a counterexample to the equivalence of $e$ and $f$, because either $e$ reaches the output packet but $f$ does not, or vice versa.
In fact, we can obtain a representation of \emph{all} counterexamples by computing the symmetric difference of the languages of the canonical automata of $e$ and $f$ for each pair of packet headers $\head_1,\head_2$.
% The symmetric difference of two languages $L_1$ and $L_2$ is the set of words that are in exactly one of the languages, i.e. $L_1 \triangle L_2 = (L_1 \setminus L_2) \cup (L_2 \setminus L_1)$.
% \begin{wrapfigure}{r}{0.5\textwidth}
%   \vspace{-1em}
%   \includegraphics[width=\linewidth]{stackat.png}
%   \caption{A screenshot of the \StacKAT decision procedure.}
%   \label{fig:demo}
%   \vspace{-2em}
% \end{wrapfigure}
The symmetric difference of the canonical automata of $e$ and $f$ is a regular language, and the corresponding automaton can be used to enumerate all counterexamples, or to find the smallest.

\subsection{Implementation}

An interactive web-based demonstration of the \StacKAT decision procedure is available online at:
\url{https://apndx.org/pub/iqe6/stackatv1.html}
The user can input two \StacKAT programs, and the decision procedure will output whether the programs are equivalent, and if not, give a counterexample input that produces different outputs for the two programs. The decision procedure relies only on the elementary constructions on finite automata, as described in this paper. The implementation differs from the decision procedure in the paper only in that it supports inequality tests ($f \neq v$) directly, without the need to encode this as a sum over the value space.

\subsubsection{Performance Evaluation}

The StacKAT equivalence checker was evaluated using five benchmark categories, each testing different aspects of the algorithm's performance:

\begin{itemize}
    \item \textbf{Header-Stack:} Tests expressions that combine header tests with stack operations, of the form $(f_1{=}0 \cdot \push{0} + f_1{=}1 \cdot \push{1}) \cdot \ldots \cdot (f_n{=}0 \cdot \push{0} + f_n{=}1 \cdot \push{1}) \cdot (\pop{0} + \pop{1}) \cdot \ldots \cdot (\pop{0} + \pop{1})$ compared with the identity $1$. % chktex 11

    \item \textbf{Nested Alternation:} Tests how deeply nested choice operations affect performance by comparing left-to-right nesting $((\push{1} + \push{2}) + \push{3})$ with right-to-left nesting $(\push{3} + (\push{2} + \push{1}))$.

    \item \textbf{Stack Depth:} Evaluates performance with expressions that push values onto the stack and then pop them off in reverse order: $\push{1} \cdot \ldots \cdot \push{n} \cdot \pop{n} \cdot \ldots \cdot \pop{1} \equiv 1$. % chktex 11

    \item \textbf{Kleene Star Nesting:} Tests nested repetition operations with expressions of the form $(\ldots ((\push{1}^*)^*)^* \ldots)^* \equiv \push{1}^*$.

    \item \textbf{Packet Header Independence:} Tests commutativity of header checks with expressions of the form $(h_1{=}v_1) \cdot (h_2{=}v_2) \cdot \ldots \cdot (h_n{=}v_n) \equiv (h_n{=}v_n) \cdot \ldots \cdot (h_2{=}v_2) \cdot (h_1{=}v_1)$. % chktex 11
\end{itemize}

\begin{figure}[t]
  \centering
  \includegraphics[width=\textwidth]{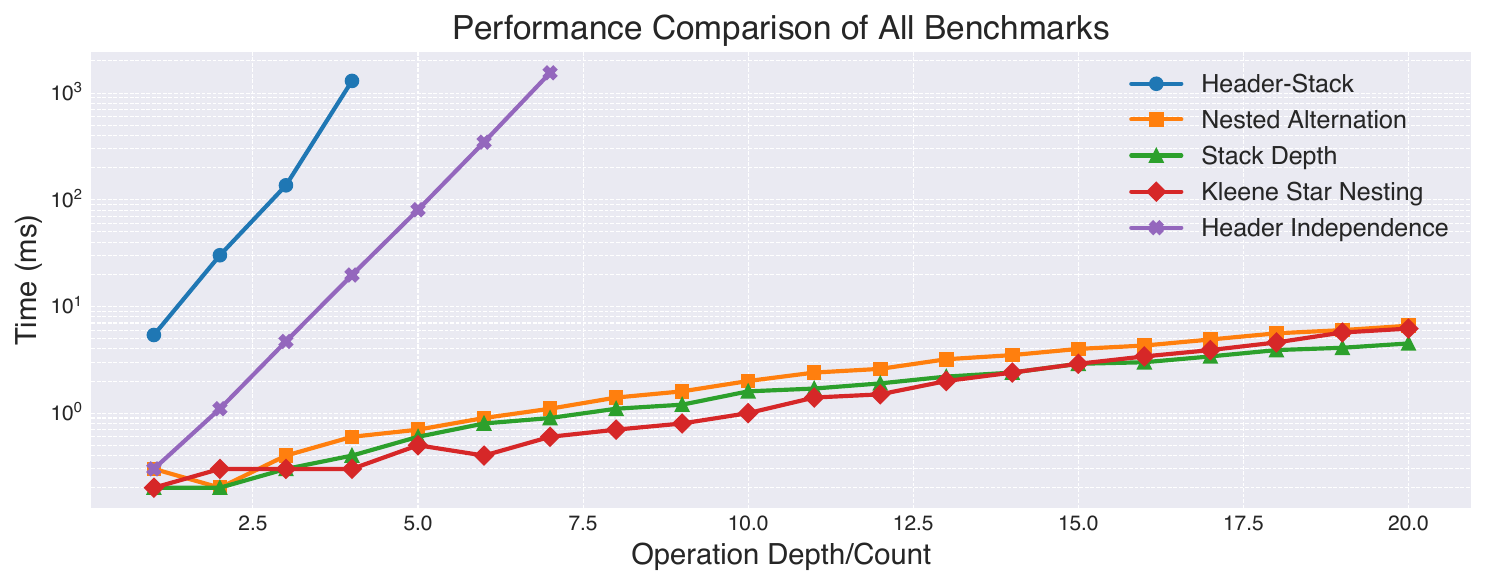}
  \caption{Performance comparison of all benchmarks. The x-axis shows the number of operations or depth for each benchmark, and the y-axis (logarithmic scale) shows execution time in milliseconds.}%
  \label{fig:all-benchmarks}
\end{figure}

As shown in Figure~\ref{fig:all-benchmarks}, the benchmarks exhibit different scaling behaviors:

\begin{itemize}
    \item \textbf{Exponential Growth:} Header-Stack and Header Independence benchmarks show the steepest increase in execution time, with both growing exponentially in the number of operations. These benchmarks involve state space explosions due to independent variables.

    \item \textbf{Moderate Growth:} Nested Alternation, Stack Depth, and Kleene Star Nesting show more moderate growth rates, making them more tractable for larger inputs.
\end{itemize}

The results indicate that the StacKAT equivalence checker performance is primarily affected by header space size, which grows exponentially with the number of independent header variables. Operations that primarily affect stack manipulation without expanding the state space show more favorable performance characteristics.

\subsubsection{Naïveté of the Implementation}

The implementation is naïve in the sense that it enumerates over the packet header space and checks equivalence for each input-output packet header pair. For each such pair, it constructs the zipped automata, and checks language equivalence of the two automata. The alphabet of the automata consists of concrete numeric values that are pushed and popped by the program. We leave the design and implementation of a \emph{symbolic} decision procedure for \StacKAT as future work. A symbolic decision procedure would not naively enumerate over the packet space, nor would it have concrete numeric values on the edges of the automaton: it would instead use BDD-like structures, which have been shown to significantly speed up equivalence checking of NetKAT programs~\cite{Moeller2024}.
A simple version would represent the alphabet on edges symbolically using a BDD, which would allow operations such as generalized pop ($\pop{V}$, where $V$ is a set of values) to be represented as a single edge.
How to best integrate symbolic methods with \StacKAT to allow efficient header-to-stack and stack-to-header transfer is a question for future research, and the current implementation should be seen as an interactive playground for exploration of the theory of \StacKAT.

\section{Axiomatization and Completeness}%
\label{sec:completeness}
\jules{We should harmonize the notation used in this section with the rest of the paper.}

In this section we propose an axiomatization of the push-pop fragment of \stackat, and prove its completeness with respect to the equational theory of $\Sem{-}$.

Our development reuses the language model developed for the decision procedure, and can be divided into two phases:
\begin{itemize}
\item
First, we axiomatize the equational theory of a simple language model corresponding to push-pop canonicalization and filtering invalid traces---i.e., $L_0(-) = \filter(\pushpop{L(-)})$. For this axiomatization, we extend Kleene Algebra (KA) with axioms $\push v \pop v=1$ and $\push{v} \pop{w}=0$ for $v\ne w$.
\item
Next, we extend this axiomatization to the equational theory of $\overline{L}$ (resp.\ $\Sem{-}$, by \Cref{thm:semcanon}), using the axiom $\pop v \push v \leq 1$, and an additional axiom scheme.
Because the interaction of push and pop exhibits non-regular, context-free behavior, we introduce a new language operator $^\dagger$, to allow the formation of a limited class of context-free languages.
\end{itemize}

For brevity, this section highlights the key aspects of our axiomatization and completeness result. For full details, please see the appendix.

\subsection{Completeness for $L_0$}

Let $L_0$ be the language model corresponding to push-pop canonicalization and filtering invalid traces. That is, $L_0(-) = \filter(\pushpop{L(-)})$. To abbreviate notation, let us write $\push{a} \in \Sigma$ simply as $a$, and $\pop{a} \in \Sigma$ as $\bar{a}$.
In this section, we show that the axioms of KA, augmented with the equations $a\bar{a}=1$ and $a\bar{b}=0$ for $a\ne b$ are complete for the equational theory of $L_0$, i.e., $L_0(e_1) = L_0(e_2)$ if $e_1 = e_2$ can be proved from these. % (Theorem \ref{thm:completenessL0}).
%that is, for any expressions $e_1,e_2$, $L_0(e_1)=L_0(e_2)$ iff $e_1=e_2$ is a deductive consequence of the axioms.
%Here and for the rest of this section, ``deductive consequence'' refers to provability in ordinary equational logic from the axioms of KA and the equations $a\bar a=1$ and $a\bar b=0$ for $b\ne a$, $a,b\in\Sigma$.

%Recall that $R:\Exp\to\AAA$ is the standard interpretation of regular expressions over the alphabet $\Sigma\cup\bar\Sigma$ and $L_0:\Exp\to\BB$ with $L_0 = \NF\circ R$. Thus $L_0$ maps $e$ to the set of strings obtained by reducing strings in $L(e)$ to normal form under the rewrite rules $a\bar a\to 1$ and $a\bar b\to 0$, $b\ne a$, discarding those that become $0$.

\smallskip
Recall that the push-pop fragment of \stackat\ is formed by regular expressions over $\Sigma$.
Let $(E,D)$ be the usual Brzozowski derivative on regular expressions~\cite{Brzozowski1964}.
Let $D'$ be like $D$, except it removes symbols from the right instead of the left.
The following two versions of the fundamental theorem for regular expressions can be proved from the laws of KA alone:~\cite{Silva2010}:
\begin{align}
& e = E(e) + \sum_{a\in\Sigma} aD_a(e)
&& e = E(e) + \sum_{a\in\Sigma} D'_a(e)a\label{eq:fundthmZ}
\end{align}
We can straightforwardly extend $D$ and $D'$ to words: for $x\in\Sigma^*$ and $a\in\Sigma$, we define
\begin{align*}
& D_\eps(e) = e
&& D_{xa}(e) = D_{a}(D_{x}(e))
&& D'_\eps(e) = e
&& D'_{xa}(e) = D'_{a}(D'_{x}(e)).
\end{align*}
%Then
%\begin{align*}
%& D_{xy}(e) = D_{y}(D_{x}(e)) && D'_{xy}(e) = D'_{y}(D'_{x}(e)).
%\end{align*}
There are only finitely many derivatives $D_{x}(e)$ and $D'_{x}(e)$ up to KA equivalence~\cite{Brzozowski1964}.
%We will apply $(E,D)$ and $(E,D')$ to both pure push and pure pop expressions below.

A \emph{pure push} (respectively, \emph{pure pop}) \emph{expression} is a regular expression over letters of the form $a \in \Sigma$ (respectively, $\bar{a} \in \Sigma$).
The expressions $1$ and $0$ are both pure pop and pure push expressions.

\begin{theoremEnd}[default]{lemma}%
\label{lem:sumpushpop}
Every expression $e_1\bar e_2$, where $e_1$ is a pure push expression and $\bar e_2$ a pure pop expression, can be provably transformed to a finite sum of pure push expressions of the form $D'_{xa}(e_1)a$, pure pop expressions of the form $\bar aD_{xa}(\bar e_2)$, or $1$, where $x\in\Sigma^*$ and $a\in\Sigma$.
\end{theoremEnd}
\begin{proofEnd}
We use the fundamental theorem \eqref{eq:fundthmZ} with $D$ applied to $e_1$ and $D'$ applied to $\bar{e}_2$: % chktex 2
\begin{align*}
e_1 &= E(e_1) + \sum_{a\in \Sigma} D'_a(e_1)a & \bar e_2 &= E(\bar e_2) + \sum_{a\in \Sigma} \bar aD_{a}(\bar e_2).
\end{align*}
We have
\begin{align*}
e_1\bar e_2 &= \Big(E(e_1) + \sum_{a\in \Sigma} D'_a(e_1)a\Big)\Big(E(\bar e_2) + \sum_{b\in \Sigma} \bar b D_b(\bar e_2)\Big)\\
&= E(e_1)E(\bar e_2) + E(\bar e_2)\sum_{a\in \Sigma} D'_a(e_1)a + E(e_1)\sum_{b\in \Sigma} \bar b D_b(\bar e_2) + \Big(\sum_{a\in \Sigma} D'_a(e_1)a\Big)\Big(\sum_{b\in\Sigma}\bar bD_b(\bar e_2)\Big) \\
&= E(e_1)\bar e_2 + E(\bar e_2)e_1 + \sum_{a\in \Sigma} D'_a(e_1)D_a(\bar e_2).
\end{align*}
In the last step, we have simplified the summation by applying the axioms $\push{a}\pop{a} = 1$ and $\push{a}\pop{b} = 0$ when $a\neq b$.
The same analysis applies to $D'_x(e_1)D_x(\bar e_2)$, giving
\begin{align*}
D'_x(e_1)D_x(\bar e_2) &= E(D'_x(e_1))D_x(\bar e_2) + E(D_x(\bar e_2))D'_x(e_1) + \sum_a D'_{xa}(e_1)D_{xa}(\bar e_2)
\end{align*}
for $x\in\Sigma^*$, or more succinctly,
\begin{align}
\alpha_x &= \phi_x + \sum_a \alpha_{xa}\label{eq:alpha}
\end{align}
for $x\in\Sigma^*$, where
\begin{align*}
\alpha_x &= D'_x(e_1)D_x(\bar e_2) & \phi_x = E(D'_x(e_1))D_x(\bar e_2) + E(D_x(\bar e_2))D'_x(e_1).
\end{align*}
There are actually only finitely many equations \eqref{eq:alpha}, as there are only finitely many derivatives of $e_1$ and $\bar e_2$ up to associativity, commutativity, and idempotence (ACI)~\cite{Brzozowski1964}. This gives rise to a finite linear system % chktex 2
\begin{align}
X_x &= \phi_x + \sum_a X_{xa},\ x\in S\label{eq:system}
\end{align}
where $S\subs\Sigma^*$ is a finite index set large enough that all equations \eqref{eq:alpha} are represented. Then \eqref{eq:system} can be reformulated as a matrix-vector equation $X = \phi + AX$, where $A$ is an $S\times S$ matrix over $\{0,1\}$ and $\phi$ is a vector of sums of pure push or pure pop expressions $\phi_x$. The equations \eqref{eq:alpha} say that $\alpha = \phi + A\alpha$, which says that $X = \alpha$ is a solution of \eqref{eq:system}. % chktex 2

Now let $\psi$ be the vector with
\begin{align*}
\psi_x &= \sum_{y\in\Sigma^*}\phi_{xy},\ x\in S.
\end{align*}
These are essentially finite sums, again because there are only finitely many derivatives up to ACI\@. We claim that the resulting finite vector $\psi$ is also a solution to \eqref{eq:system}, and is in fact the least solution. To see that it is a solution, % chktex 2
\begin{align*}
\psi_x &= \sum_y \phi_{xy}
= \phi_x + \sum_{a} \sum_{y} \phi_{xay}
= \phi_x + \sum_a \psi_{xa}.
\end{align*}
It is also the least, since one can show by induction that for any other solution $X$ and any $x$, $\phi_{xy} \le X_x$ for all $y$:
\begin{align*}
\phi_x &\le \phi_x + \sum_a X_{xa} = X_x
&
\phi_{xay} &\le X_{xa}
\le \phi_x + \sum_a X_{xa} = X_x.
\end{align*}
By Kleene algebra~\cite{Kozen1994}, the least solution $\psi$ can be calculated efficiently as $\psi = A^*\phi$, which is a vector of sums of pure push and pure pop expressions.

Let
\begin{align*}
\psi(e_1\bar e_2) &= \psi_\eps = \sum_{x\in\Sigma^*} E(D'_x(e_1))D_x(\bar e_2) + E(D_x(\bar e_2))D'_x(e_1).
\end{align*}
Since $\alpha$ is a solution of \eqref{eq:system}, we know that $\psi(e_1\bar e_2)\le\alpha_\eps=e_1\bar e_2$. % chktex 2
We claim also that $e_1\bar e_2\le\psi(e_1\bar e_2)$. This will follow from several sublemmas.
\begin{lemma}%
\label{lemma:psi-monotone}
If $f_1 \leq g_1$ and $\bar f_2 \leq \bar g_2$, then $\psi(f_1 \bar f_2) \leq \psi(g_1 \bar g_2)$.
\end{lemma}
\begin{proof}
Using monotonicity of the derivatives,
\begin{align*}
\psi(f_1\bar f_2)
&= \sum_{x\in\Sigma^*} E(D'_x(f_1))D_x(\bar f_2) + E(D_x(\bar f_2))D'_x(f_1)\\
&\le \sum_{x\in\Sigma^*} E(D'_x(g_1))D_x(\bar g_2) + E(D_x(\bar g_2))D'_x(g_1)
= \psi(g_1\bar g_2).
\tag*\qedhere
\end{align*}
\end{proof}

\begin{lemma}%
\label{lemma:insert-left-base}
$a\psi(e_1 \bar e_2) \leq \psi(a e_1 \bar e_2)$.
\end{lemma}
\begin{proof}
We wish to show
\begin{align*}
\sum_{x\in\Sigma^*} E(D'_x(e_1))aD_x(\bar e_2) + E(D_x(\bar e_2))aD'_x(e_1)
\le \sum_{x\in\Sigma^*} E(D'_x(ae_1))D_x(\bar e_2) + E(D_x(\bar e_2))D'_x(ae_1).
\end{align*}
It suffices to show the two inequalities
\begin{align*}
aD'_x(e_1) &\le D'_x(ae_1) &
E(D'_x(e_1))aD_x(\bar e_2)
&\le E(D'_{xa}(ae_1))D_{xa}(\bar e_2) + E(D_x(\bar e_2))D'_x(ae_1).
\end{align*}
The proof of the left-hand inequality is by induction on $x$. The basis is $aD'_\eps(e_1) = ae_1 = D'_\eps(ae_1)$. For the induction step,
\begin{align*}
aD'_{xb}(e_1)
&= aD'_{b}(D'_{x}(e_1))
\le D'_{b}(aD'_{x}(e_1))
\le D'_{b}(D'_{x}(ae_1))
= D'_{xb}(ae_1).
\end{align*}
For the right-hand inequality, if $E(D'_x(e_1))=0$, we are done. Otherwise, $E(D'_x(e_1))=1$, thus $x^R\in e_1$, thus $(xa)^R = ax^R\in ae_1$, thus also $E(D'_{xa}(ae_1))=1$ and $a\le D'_x(ae_1)$. We now must show
\begin{align*}
aD_x(\bar e_2)
&\le D_{xa}(\bar e_2) + E(D_x(\bar e_2))D'_x(ae_1).
\end{align*}
Using the fundamental theorem,
\begin{align*}
aD_x(\bar e_2)
&= a(E(D_x(\bar e_2)) + \sum_b \bar b D_{xb}(\bar e_2))
= E(D_x(\bar e_2))a + D_{xa}(\bar e_2)\\
&\le E(D_x(\bar e_2))D'_x(ae_1) + D_{xa}(\bar e_2).
\tag*\qedhere
\end{align*}
\end{proof}

\begin{lemma}%
\label{lemma:insert-left}
For all pure push expressions $e_1$ and $d$ and all pure pop expressions $\bar e_2$, it is provable that $d \psi(e_1 \bar e_2) \leq \psi(de_1 \bar e_2)$.
\end{lemma}
\begin{proof}
We proceed by induction on $d$. For the basis, the cases $\{0,1\}$ are trivial, and the case $a \in \Sigma$ follows from Lemma~\ref{lemma:insert-left-base}. For the inductive case $d_1 + d_2$,
    \begin{align*}
    (d_1 + d_2)\psi(e_1 \bar e_2)
        &= d_1\psi(e_1 \bar e_2) + d_2\psi(e_1 \bar e_2) \\
        &\leq \psi(d_1e_1 \bar e_2) + \psi(d_2e_1 \bar e_2) && \text{induction hypothesis}\\
        &\leq \psi((d_1 + d_2)e_1 \bar e_2) && \text{\Cref{lemma:psi-monotone}.}
    \end{align*}

The case $d_1d_2$ follows from two applications of the induction hypothesis:
    \begin{align*}
    d_1d_2\psi(e_1 \bar e_2)
        \leq d_1\psi(d_2e_1 \bar e_2)
        \leq \psi(d_1d_2e_1\bar e_2).
    \end{align*}

Finally, for the case $d^*$,
    \begin{align*}
    \psi(e_1 \bar e_2) + d\psi(d^*e_1 \bar e_2)
        &\leq \psi(e_1 \bar e_2) + \psi(dd^* e_1 \bar e_2) && \text{induction hypothesis}\\
        &\leq \psi(d^* e_1 \bar e_2)  && \text{\Cref{lemma:psi-monotone}.}
    \end{align*}
    By an axiom of Kleene algebra, $d^* \psi(e_1 \bar e_2) \leq \psi(d^*e_1 \bar e_2)$.
\end{proof}

\begin{lemma}%
\label{lemma:insert-right}
For all pure push expressions $e_1$ and all pure pop expressions $\bar e_2$ and $\bar d$, it is provable that $\psi(e_1 \bar e_2) \bar d \leq \psi(e_1 \bar e_2 \bar d)$.
\end{lemma}
\begin{proof}
By symmetric arguments to \Cref{lemma:insert-left}.
\end{proof}

\begin{lemma}
For all pure push expressions $e_1$ and all pure pop expressions $\bar e_2$, $e_1 \bar e_2 \leq \psi(e_1 \bar e_2)$.
\end{lemma}
\begin{proof}
Using \Cref{lemma:insert-left,lemma:insert-right} and the fact that $1\leq\psi(1)$,
$e_1\bar e_2 \leq e_1\psi(1)\bar e_2 \leq \psi(e_1 \bar e_2)$.
\end{proof}

We have shown that
\begin{align*}
e_1\bar e_2 &= \sum_{x\in\Sigma^*} E(D'_x(e_1))D_x(\bar e_2) + E(D_x(\bar e_2))D'_x(e_1).
\end{align*}
To obtain the form as specified in the statement of the lemma (Lemma~\ref{lem:sumpushpop}), we can simply write
$D_x(\bar e_2)$ as $E(D_x(\bar e_2)) + \sum_a \bar aD_{xa}(\bar e_2)$ and
$D'_x(e_1)$ as $E(D'_x(e_1)) + \sum_a D_{xa}(e_1)a$ using the fundamental theorem \eqref{eq:fundthmZ}. % chktex 2
\end{proofEnd}

An expression is in \emph{normal form} if it is a sum of expressions of the form $\bar e_1e_2$, where $\bar e_1$ is a pure pop expression and $e_2$ is a pure push expression. Note that $L(e)=L_0(e)$ if $e$ is in normal form.

\begin{theoremEnd}[default]{theorem}%
\label{thm:normalform}
Every expression can be provably transformed to normal form.
\end{theoremEnd}
\begin{proofEnd}
By induction. For the basis, $a$, $\bar a$, $1$, and $0$ are all in normal form.
For $+$, we transform both terms of the sum individually to normal form; clearly, normal forms are preserved by sums.

For a product of expressions with each factor already in normal form, we multiply out the sums using distributivity.
This gives us a big sum of expressions of the form $\bar{e}_1e_2\bar{e}_1'e_2'$, with $\bar{e}_1$ and $\bar{e}_1'$ pure pop expressions, while $\bar{e}_2$ and $\bar{e}_2'$ are pure push expressions.
It now suffices to show how to transform $\bar{e}_1e_2\bar{e}_1'e_2'$ into normal form.
We do this using \Cref{lem:sumpushpop}, which tells us that the inner expression $e_2\bar{e}_1'$ can be written as $\bar{d}_1 + d_2$, with $\bar{d}_1$ a pure pop expression, and $d_2$ a pure push expression.
This lets us achieve the desired normal form, by deriving as follows:
\begin{align*}
\bar e_1e_2\bar e_1'e_2' = \bar e_1(e_2\bar e_1')e_2' = \bar e_1(\bar d_1 + d_2)e_2' = \bar e_1\bar d_1e_2' + \bar e_1d_2e_2'.
\end{align*}

Finally, consider $e^*$, where $e = \bar b_1c_1 + \cdots + \bar b_n c_n$ has already been transformed to normal form.
Let $[\bar b]$ be the vector whose $i$th component is $\bar b_i$ and let $[c]$ be the vector whose $i$th component is $c_i$.
Then $e = [\bar b]^T[c]$, where $^T$ denotes matrix transpose.
Using basic Kleene algebra, we then derive:
\begin{align}
\Big(\sum_i \bar b_i c_i\Big)^* &= ([\bar b]^T[c])^*
= 1 + [\bar b]^T[c]([\bar b]^T[c])^*
= 1 + [\bar b]^T([c][\bar b]^T)^*[c]
= 1 + [\bar b]^T A^*[c],\label{eq:Amatrix}
\end{align}
where $A = [c][\bar b]^T$ is an $n\times n$ matrix with $A_{ij}=c_i\bar b_j$.
By applying \Cref{lem:sumpushpop} to the terms in $A$, we can express this matrix as $B_0 + C_0$, where $(B_0)_{ij}$ contains only pure push terms $D'_{xa}(c_i)a$ (in general different terms for different $j$) or $1$ and $(C_0)_{ij}$ contains only pure pop terms $\bar aD_{xa}(\bar b_j)$ or $1$.

Now consider $B_0C_0$.
Each entry $(B_0C_0)_{ij}$ is a sum of terms of the form either $D'_{xa}(c_i)a$, $\bar aD_{xa}(\bar b_j)$, $D'_{xa}(c_i)a\bar cD_{xc}(\bar b_j)$ or $1$. We can delete $D'_{xa}(c_i)a\bar cD_{xc}(c_j)$ when $a\ne c$.
If $a=c$, we can replace $D'_{xa}(c_i)a\bar aD_{xa}(\bar b_j)$ by $D'_{xa}(c_i)D_{xa}(\bar b_j)$, then replace this by a sum of pure push and pure pop terms according to Lemma~\ref{lem:sumpushpop}.
The new terms will still be of the form $D'_{xa}(c_i)a$ or $\bar aD_{xa}(\bar b_j)$. We add these new terms to $B_0$ and $C_0$ to obtain matrices $B_1$ and $C_1$ respectively.
Then $B_1 + C_1 = B_0 + C_0 + B_0C_0$.

We continue this process, getting $B_0\le B_1\le B_2\le\cdots$ and $C_0\le C_1\le C_2\le\cdots$ as long as there are new terms. Eventually we do not get any new terms because there are only finitely many derivatives up to ACI\@. At that point $B_{n+1}=B_n$ and $C_{n+1}=C_n$, and $B_n C_n\le B_{n+1}+C_{n+1} = B_n+C_n$.

We claim at this point that the original $A$ of \eqref{eq:Amatrix} satisfies % chktex 2
\begin{align*}
A^* &= (B_n + C_n)^* = C_n^*B_n^*,
\end{align*}
where the right-hand side is of the desired form.
For the left-hand equality, surely $A^* = (B_0 + C_0)^* \le (B_n + C_n)^*$, and by induction, $B_{i+1}+C_{i+1} = B_i + C_i + B_i C_i\le A^* + A^* + A^*A^* = A^*$, so $(B_n+C_n)^*\le A^{**} = A^*$. For the right-hand equality, it is clear that $C_n^*B_n^* \le (B_n + C_n)^*$, so it remains to show that $(B_n + C_n)^* \le C_n^*B_n^*$. Since $B_n C_n \le B_n + C_n$, we have
\begin{align*}
B_n + (B_n + C_n^*)C_n = B_n + B_n C_n + C_n^*C_n \le B_n + B_n + C_n + C_n^* \le B_n + C_n^*.
\end{align*}
Applying the left-hand fixpoint rule $b + xa \le x\ \Imp\ ba^* \le x$ of KA with $b = B_n$, $a = C_n$, and $x = B_n + C_n^*$, we have that $B_n C_n^*\le B_n + C_n^*$. Then
\begin{align*}
1 + (B_n + C_n)C_n^*B_n^* &= 1 + B_n C_n^*B_n^* + C_n C_n^*B_n^* \le 1 + (B_n + C_n^*)B_n^* + C_n C_n^*B_n^* \le C_n^*B_n^*.
\end{align*}
Now applying the right-hand fixpoint rule $b + ax \le x\ \Imp\ a^*b \le x$ of KA with $a = B_n + C_n$, $b=1$, and $x = C_n^*B_n^*$, we conclude that $(B_n + C_n)^* \le C_n^*B_n^*$.
\end{proofEnd}

We now have everything we need to state our completeness result for $L_0$.

\begin{theorem}%
\label{thm:completenessL0}
The axioms of KA augmented with the equations $a\bar a=1$ and $a\bar b=0$ for $b\ne a$ are sound and complete for the equational theory of the interpretation $L_0$; that is, for any $e_1,e_2\in\Exp$, $L_0(e_1)=L_0(e_2)$ iff $e_1=e_2$ is a deductive consequence of these axioms.
\end{theorem}
\begin{proof}
Soundness of the axioms is routine; thus if $e_1=e_2$ is provable, then $L_0(e_1)=L_0(e_2)$.
Conversely, suppose $L_0(e_1)=L_0(e_2)$.
Using \Cref{thm:normalform}, we can provably transform $e_1$ and $e_2$ to normal forms $e_1'$ and $e_2'$, respectively.
By soundness, $L_0(e_1')=L_0(e_2')$, and since $L_0(e)=L(e)$ for any $e$ in normal form (by \Cref{lem:sumpushpop}), $L(e_1')=L(e_2')$.
By the completeness of KA w.r.t. $L$~\cite{Kozen1994}, $e_1'=e_2'$ is provable.
In conjunction with the proofs of $e_1=e_1'$ and $e_2=e_2'$, we obtain a proof of $e_1=e_2$.
\end{proof}

\subsection{Completeness for \stackat}

We now turn our attention to proving completeness for the push-pop fragment of \StacKAT with respect to $\bar{L}$.
We continue to assume the axioms for $L_0$, namely the axioms of KA in addition to the equations $a\bar a=1$ and $a\bar b=0$ for $b\ne a$, but $\bar{L}$ (and therefore $\Sem{-}$) also satisfy an additional property:

\begin{theoremEnd}[default]{lemma}%
\label{lem:L0L1}
The axioms for $L_0$ plus $\bar aa \le 1$ (equivalently, $1+\bar aa = 1$) are sound for $\overline{L}$.
\end{theoremEnd}
\begin{proofEnd}
The axioms of KA are sound since $\BB$ is a KA\@. The axioms $a\bar a=1$ and $a\bar b=0$, $b\ne a$, are still sound since equations are preserved under a homomorphism. The axiom $\bar aa \le 1$, equivalently $1+\bar aa = 1$, is also sound under $L_1$:
\begin{align*}
& L_1(1+\bar aa) = \NF(\{\eps,\bar aa\})\dn = \{\eps,\bar aa\}\dn = \set{\bar xx}{x\in\Sigma^*} = \{\eps\}\dn = \NF(\{\eps\})\dn = L_1(1).
\end{align*}
Thus if $e_1=e_2$ is provable, then $L_1(e_1)=L_1(e_2)$.
\end{proofEnd}

Although inspired by this axiom, our axiomatization is rather more complicated, involving an infinite set of rules derived from the fundamental theorem. % \eqref{eq:fundthmA}.
We do not know whether $\bar aa\le 1$ by itself is enough for completeness, but conjecture that it is not.

To properly describe our axiomatization, we need a new operator.
For $B\subs\Sigma^*$, define $B^\dagger = \set{\bar xx}{x\in B}$.
For example, $(a^*)^\dagger = \set{\bar a^n a^n}{n\ge 0}$.
Clearly, ${}^\dagger$ does not preserve regularity, which is not surprising as we are dealing with stacks.
However, we will temper the non-regular aspect of this operator by using the two Brzozowski derivatives $D$ and $D'$ in tandem. %, as we did in the proof of \Cref{lem:sumpushpop}.
Still assuming $a\bar a=1$ and $a\bar b=0$ for $b\ne a$, here are some properties of $^\dagger$ that are not difficult to derive:
%\begin{itemize}
%\item
%$(\bigcup_\alpha B_\alpha)^\dagger = \bigcup_\alpha B_\alpha^\dagger$
%\item
%$B^\dagger B^\dagger = B^\dagger$
%\item
%$(B^\dagger)^* = 1 + B^\dagger$
%\item
%$\bar aB^\dagger a = (Ba)^\dagger$, $a\in\Sigma$
%\end{itemize}
\begin{align}
& \Bigl(\bigcup_\alpha B_\alpha\Bigr)^\dagger = \bigcup_\alpha B_\alpha^\dagger
&& B^\dagger B^\dagger = B^\dagger
&& (B^\dagger)^* = 1 + B^\dagger
&& \bar aB^\dagger a = (Ba)^\dagger,\ a\in\Sigma
\label{eq:sounddaggereqs}
\end{align}
It is also clear that if $B$ is regular, then $B^\dagger$ is context-free. %, as $B^\dagger$ is the intersection of the regular set $\bar\Sigma^*B$ with the context-free language $(\Sigma^*)^\dagger = \set{\bar xx}{x\in\Sigma^*}$.
We furthermore have the following.

%We can also introduce a corresponding syntactic operator $^\dagger$, but for syntactic closure we must also extend the operator $\barnone$ to apply to pop symbols as well as push symbols. For this purpose it makes sense to define $\bar{\bar a}=a$, so that $\barnone$ becomes an involution. On $(\Sigma\cup\bar\Sigma)^*$, $\barnone$ reverses the string and simultaneously complements each symbol. However, this extension introduces more complication than necessary for the completeness of \stackat, as we will only ever need to apply $^\dagger$ to pure push expressions.
%
%With the extended definition of $\barnone$,

\begin{textAtEnd}
We can extend the interpretation $R:\Exp\to 2^{(\Sigma\cup\bar\Sigma)^*}$ to expressions involving $^\dagger$ by adding the clause $\SemR{e^\dagger} = \SemR e^\dagger$ to the inductive definition. The semantic properties \eqref{eq:sounddaggereqs} lead to sound equations % chktex 2
%\begin{itemize}
%\item
%$(e_1+e_2)^\dagger = e_1^\dagger + e_2^\dagger$
%\item
%$e^\dagger e^\dagger = e^{\dagger\dagger} = e^\dagger$
%\item
%$(e^\dagger)^* = 1 + e^\dagger$
%\item
%$\bar ae^\dagger a = (ea)^\dagger$, $a\in\Sigma$
%\end{itemize}
\begin{align*}
& (e_1+e_2)^\dagger = e_1^\dagger + e_2^\dagger
&& e^\dagger e^\dagger = e^{\dagger\dagger} = e^\dagger
&& (e^\dagger)^* = 1 + e^\dagger
&& \bar ae^\dagger a = (ea)^\dagger,\ a\in\Sigma
\end{align*}
The congruence rule $e_1=e_2\Imp e_1^\dagger=e_2^\dagger$ also holds.

From \eqref{eq:fundthmZ}, we have a version of the fundamental theorem that holds in the presence of $^\dagger$: % chktex 2
\begin{align}
\bar e_1e^\dagger e_2
%&= (E(\bar e_1) + \sum_{a\in\Sigma} D'_a(\bar e_1)\bar a)e^\dagger(E(e_2) + \sum_{b\in\Sigma} bD_b(e_2))\nonumber\\
&= E(\bar e_1)e^\dagger E(e_2) + E(\bar e_1)\sum_b e^\dagger bD_{b}(e_2)\nonumber + E(e_2)\sum_a D'_{a}(\bar e_1)\bar ae^\dagger\nonumber\\
&\quad\ + \sum_{a\ne b} D'_{a}(\bar e_1)\bar a e^\dagger bD_{b}(e_2) + \sum_a D'_{a}(\bar e_1)(ea)^\dagger D_{a}(e_2).\label{eq:fundthmA}
\end{align}
We also have $e^\dagger = E(e) + \sum_{a\in\Sigma} \bar aD'_a(e)^\dagger a$.
\end{textAtEnd}

\begin{theoremEnd}[default]{lemma}%
\label{lem:LR}
For pure pop expressions $\bar{e}_1$ and pure push expressions $e_2$,
$\bar{L}(\bar e_1e_2) = L(\bar e_1) \cdot \Ssd \cdot L(e_2)$.
\end{theoremEnd}
\begin{proofEnd}
By definition, $\overline{L}(\bar e_{1}e_{2}) = \poppush{L_0(\bar e_{1}e_{2})} = \poppush{L(\bar e_{1}e_{2})} = L(\bar e_{1}) \cdot \Ssd \cdot L(e_{2})$.
\end{proofEnd}

We can also extend ${}^\dagger$ to regular expressions, interpreting $e^\dagger$ as $L(e)^\dagger$.
The following then holds.

\begin{theoremEnd}[default]{lemma}%
\label{lem:delimited}
Let $\bar e_1$ be a pure pop expression and $e,e_2$ pure push expressions.
The expression $\bar e_1e^\dagger e_2$ is semantically $L$-equivalent to a sum of expressions of the following forms:
\begin{align}
& \bar d_1\bar ad^\dagger bd_2,\ a\ne b && \bar d_1\bar ad^\dagger && d^\dagger bd_2 && d^\dagger\label{eq:four1}
\end{align}
where $d_1$ is a pure pop expression, $d$ and $d_2$ are pure push expressions, and $a,b\in\Sigma$, $a\ne b$.
\end{theoremEnd}
\begin{proofEnd}
Fix $\bar e_1$, $e$, and $e_2$ as in the statement of the lemma. Define a binary relation on $\Sigma^*$:
\begin{align*}
y\equiv z\ \ \Iff\ \ \SemR{D'_{y}(\bar e_1)} = \SemR{D'_{z}(\bar e_1)}\ \ \text{and}\ \ \SemR{D_{y}(e_2)} = \SemR{D_{z}(e_2)}.
\end{align*}
The relation $\equiv$ is a right-invariant equivalence relation (that is, $y\equiv z$ implies $ya\equiv za$) with only finitely many classes $[z]=\set y{y\equiv z}$, as there are only finitely many derivatives up to KA equivalence. Thus there is a finite-state automaton with states $\Sigma^*/{\equiv}$ and transition function $\delta([x],a)=[xa]$. The equivalence classes $[z]$ are all regular sets and are the strings accepted from the state $[z]$.

One can compute a regular expression $e_z$ for $[z]$ as follows. The single-step transition matrix of this automaton is a $(\Sigma^*/{\equiv})\times(\Sigma^*/{\equiv})$ matrix $A$ of expressions such that $A_{[x][z]} = \sum_{[xa]=[z]} a$. Thus
\begin{align*}
\SemR{A_{[x][z]}} = \set{a\in\Sigma}{[xa]=[z]}.
\end{align*}
It follows by induction that
\begin{align*}
& \SemR{A^n_{[x][z]}} = \set{y\in\Sigma^n}{[xy]=[z]}
&& \SemR{A^*_{[x][z]}} = \set{y\in\Sigma^*}{[xy]=[z]}
\end{align*}
so we can take $e_z = A^*_{[\eps][z]} = [z]$.

Now for $x\in\Sigma^*$, define
\begin{align*}
\tau(x)
&= D'_x(\bar e_1)(ex)^\dagger D_x(e_2)\\
\phi(x)
&= E(D'_x(\bar e_1))(ex)^\dagger E(D_x(e_2)) + E(D'_x(\bar e_1))\sum_b (ex)^\dagger bD_{xb}(e_2)\\
&\quad\ + E(D_x(e_2))\sum_a D'_{xa}(\bar e_1)\bar a(ex)^\dagger + \sum_{a\ne b} D'_{xa}(\bar e_1)\bar a(ex)^\dagger bD_{xb}(e_2).
\end{align*}
Using the fact that $\sum_{y\equiv z}(ey)^\dagger = (ee_z)^\dagger$ and the fact that $D'_{y}(\bar e_1)$ and $D'_{z}(\bar e_1)$ are equivalent and $D_{y}(e_2)$ and $D_{z}(e_2)$ are equivalent for $y\equiv z$, we have
\begin{align}
\sum_{y\equiv z}\phi(y)
&= \sum_{y\equiv z}E(D'_z(\bar e_1))(ey)^\dagger E(D_z(e_2)) + \sum_{y\equiv z}E(D'_z(\bar e_1))\sum_b (ey)^\dagger bD_{zb}(e_2)\nonumber\\
&\quad\ + \sum_{y\equiv z}E(D_z(e_2))\sum_a D'_{za}(\bar e_1)\bar a(ey)^\dagger + \sum_{y\equiv z}\sum_{a\ne b} D'_{za}(\bar e_1)\bar a(ey)^\dagger bD_{zb}(e_2)\nonumber\\
&= E(D'_y(\bar e_1))(ee_z)^\dagger E(D_y(e_2)) + E(D'_y(\bar e_1))\sum_b (ee_z)^\dagger bD_{yb}(e_2)\nonumber\\
&\quad\ + E(D_y(e_2))\sum_a D'_{ya}(\bar e_1)\bar a(ee_z)^\dagger + \sum_{a\ne b} D'_{ya}(\bar e_1)\bar a(ee_z)^\dagger bD_{yb}(e_2).\label{eq:goodform}
\end{align}
This is a finite sum of terms of the desired form as described in the statement of the lemma.

The fundamental theorem \eqref{eq:fundthmA} says that $\tau(x) = \phi(x) + \sum_a \tau(xa)$. Unwinding this inductively, % chktex 2
\begin{align*}
\bar e_1e^\dagger e_2 &= \tau(\eps) = \sum_y \phi(y) = \sum_{[z]} \sum_{y\equiv z} \phi(y),
\end{align*}
which by \eqref{eq:goodform} is equivalent to a finite sum of the desired form. % chktex 2
\end{proofEnd}

Note that the four forms shown in \eqref{eq:four1} represent pairwise disjoint sets under $L$. % chktex 2
For example, every string in $L(\bar d_1\bar ad^\dagger bd_2)$, where $a\ne b$, contains a substring $\bar a\kern1pt\bar xxb$, whereas every string in $L(\bar d_1\bar ad^\dagger)$ has a suffix of the form $\bar a\kern1pt\bar xx$, and no string in normal form can have both.

%For expressions in one of the four forms \eqref{eq:four1},
Let $\#$ be a new delimiter symbol and consider the forms
\begin{align}
& ad_1\#d\#bd_2,\ a\ne b && ad_1\#d\# && \#d\#bd_2 && \#d\#\label{eq:four2}
\end{align}
corresponding to the four forms of \eqref{eq:four1}. % chktex 2
% The $\#$ are used to delimit $d$, and occurrences of $^\dagger$ are removed.
Define the map
\begin{align*}
& K:\Sigma^*\#\Sigma^*\#\Sigma^* \to \bar\Sigma^*\Sigma^* && K(x\#y\#z) = \bar x\bar yyz.
\end{align*}
$K$ is not injective; for example, $K(a\#a\#a) = K(aa\#\#aa)$.
However, $K$ acts bijectively between strings of the form~\eqref{eq:four1} and corresponding strings of the form~\eqref{eq:four2}.
\begin{theoremEnd}[default]{lemma}%
\label{lem:LR2}
Let $e$ be one of the expressions in \eqref{eq:four1} and let $e'$ be the corresponding expression in \eqref{eq:four2}. The map $K$ restricted to $L(e')$ is a bijection $K:L(e')\to L(e)$. % chktex 2
\end{theoremEnd}
\begin{proofEnd}
Consider the first two corresponding expressions $\bar d_1\bar ad^\dagger bd_2$ from \eqref{eq:four1} and $ad_1\#d\#bd_2$ from \eqref{eq:four2}, where $a\ne b$. The image of $L(ad_1\#d\#bd_2)$ under $K$ i % chktex 2
\begin{align*}
\set{K(u)}{u\in L(ad_1\#d\#bd_2)}
&= \set{K(ax\#y\#bz)}{x\in L(d_1),\ y\in L(d),\ z\in L(d_2)}\\
&= \set{\bar x\kern1pt\bar a\kern1pt\bar yybz}{x\in L(d_1),\ y\in L(d),\ z\in L(d_2)}\\
&\subs L(\bar d_1\kern1pt\bar ad^\dagger bd_2).
\end{align*}
But any string of the form $\bar x\kern1pt\bar a\kern1pt\bar yybz\in\bar d_1\kern1pt\bar ad^\dagger bd_2$ is the image under $K$ of exactly one string in $L(ad_1\#d\#bd_2)$, namely $ax\#y\#bz$. The string $y$ is uniquely determined since $\bar yy$ is the longest skew-palindromic substring of $\bar x\kern1pt\bar a\kern1pt\bar yybz$ as it is delimited by $\bar a$ and $b$ and $a\ne b$, and thereafter $x$ and $z$ are also uniquely determined.

A similar argument holds for the other three cases.
\end{proofEnd}

Let $\sum_i\bar e_{i1}e_{i2}$ and $\sum_j\bar d_{j1}d_{j2}$ be two expressions in normal form. Using Lemma~\ref{lem:delimited}, reduce $\sum_i\bar e_{i1}\Ssd e_{i2}$ and $\sum_j\bar d_{j1}\Ssd d_{j2}$ to sums of expressions $\sum_i\bar p_i q_i^\dagger r_i$ and $\sum_j \bar u_j v_j^\dagger w_j$, respectively, of the form \eqref{eq:four1}. Let $\sum_i p_i\#q_i\#r_i$ and $\sum_j u_j\#v_j\#w_j$ be the corresponding sums of expressions of the form \eqref{eq:four2}. % chktex 2

\begin{theoremEnd}[default]{lemma}%
\label{lem:FG5}
$\overline{L}(\sum_i\bar e_{i1}e_{i2})\subs \overline{L}(\sum_j\bar d_{j1}d_{j2})
\ \Iff\ L(\sum_i p_i\#q_i\#r_i)\subs L(\sum_j u_j\#v_j\#w_j)$.
\end{theoremEnd}
\begin{proofEnd}
We show that the following are equivalent:
\begin{enumerate}[(i)]
\item
$L_1(\sum_i\bar e_{i1}e_{i2})\subs L_1(\sum_j\bar d_{j1}d_{j2})$
\item
$L(\sum_i\bar e_{i1}\Ssd e_{i2})\subs L(\sum_j\bar d_{j1}\Ssd d_{j2})$
\item
$L(\sum_i p_i\#q_i\#r_i)\subs L(\sum_j u_j\#v_j\#w_j)$.
\end{enumerate}
Statements (i) and (ii) are equivalent by Lemma~\ref{lem:LR} and (ii) and (iii) are equivalent by Lemma~\ref{lem:LR2}.
\end{proofEnd}

Now consider the rule
\begin{align}
\frac{\sum_i p_i\#q_i\#r_i \le \sum_j u_j\#v_j\#w_j}{\sum_i\bar e_{i1}e_{i2} \le \sum_j\bar d_{j1}d_{j2}},\label{eq:badrule2}
\end{align}
where the expressions are as described in the paragraph above Lemma~\ref{lem:FG5}.

\begin{theorem}%
\label{thm:completenessB}
The rule~\eqref{eq:badrule2}, along with $a\bar a=1$, $a\bar b=0$ for $b\ne a$, and the axioms of KA and equational logic, are sound and complete for the equational theory of the push-pop fragment of \stackat.
\end{theorem}
\begin{proof}
Suppose we want to prove that $e\le d$ under the interpretation $L_1$ or $\Sem-$. By Theorem~\ref{thm:normalform}, we can provably convert $e$ and $d$ to normal form $\sum_i\bar e_{i1}e_{i2}$ and $\sum_j\bar d_{j1}d_{j2}$, respectively. By Lemma~\ref{lem:FG5} and rule~\eqref{eq:badrule2}, to prove that $\sum_i\bar e_{i1}e_{i2} \le \sum_j\bar d_{j1}d_{j2}$ under the interpretation $L_1$, it suffices to prove $\sum_i p_i\#q_i\#r_i\le \sum_j u_j\#v_j\#w_j$ under the interpretation $R$. This is provable, as KA is complete for $R$.
\end{proof}

\section{Related Work}\label{sec:related}

In relating \StacKAT to existing extensions of regular languages it is important to make a few observations. Most importantly, a \StacKAT program does not have a separate input tape. Hence, it does not capture behaviors that require both processing the input and pushing/popping symbols to/from the stack. Using standard pushdown automata, one can model behaviors like ``after every $a$-transition, write symbol $x$ to the stack'' but this is not possible in \StacKAT. If \StacKAT programs were equivalent to automata with both an input tape and a stack, equivalence of \StacKAT languages would clearly be undecidable---e.g., we could then use the stack to parse a context-free grammar on the input tape. Having said that, there are several tractable fragments of context-free languages that are somewhat reminiscent of \StacKAT. The reader might wonder if \StacKAT is a particular instance of one of these previously studied, tractable fragments. To the best of our knowledge this is not the case. Below, we review below the most closely related work and explain the key differences.

\paragraph{Network Reachability Verification of MPLS Networks}
Closely related to our work is the verification of reachability in MPLS networks~\cite{Jensen2020,Jensen2021, Jensen2018}. Since MPLS networks use the packet as a stack, the verification problem can be cast as a reachability problem in a pushdown system. These works can also handle quantitative properties such as failures or latency. Recently, a symbolic decision procedure was developed~\cite{Beckett2022} that can efficiently handle large packet spaces.
In contrast with these works, \StacKAT handles equivalence queries. An interesting question for future work is whether ideas from the above works can be adapted to \StacKAT.

\paragraph{Visibly Pushdown Languages}
% \tobias{On page 21, I am a bit confused about ``StacKAT languages are not a subset of VPLs, as StacKAT becomes undecidable when we have a separate stack and input tape''. It is not clear to me how the former follows from the latter. Like, if this were the case, would it mean that StacKAT languages are a subset of VPLs? Also, what are StacKAT languages?}
%
One important subclass of languages that is particularly relevant for program analysis are visibly pushdown languages (VPLs)~\cite{Alur2004}.  VPLs are accepted by visibly pushdown automata. Whereas in general pushdown automata the stack operations depend on both the input symbol and the current state, in VPAs the input alphabet is partitioned in a way that the type of stack operation (push or pop) is determined by the input symbol itself alone. VPLs are more tractable than general context-free languages as they enjoy several closure and decidability properties. VPLs are closed under union, intersection, and complementation. Moreover, emptiness, membership, and language equality and inclusion are decidable.

We view the relationship between \StacKAT and VPLs as follows.
Naively, one might hope to use equivalence of pushdown automata as a basis for decision procedures.
Unfortunately, that equivalence problem is undecidable.
VPLs address this issue by restricting the pushdown automaton to a specific subclass.
\StacKAT takes a different approach and instead removes the input tape, but keeps the pushdown stack completely unrestricted.
Therefore, VPLs and \StacKAT are incomparable, being two different ways of taming the expressive power of pushdown automata.

% \tobias{I recently heard about operator precedence languages, a different kind of strictly context-free language that does have nice decidability properties (\url{https://doi.org/10.1007/978-3-031-65627-9_19}). I do not know how these compare to StacKAT though\ldots}

\paragraph{Dyck Languages}
% \tobias{The bibliography on Dyck languages is quite extensive, maybe bigger than it needs to be?}
%
VPLs are more general than so-called {\em parentheses languages}, also often referred to as Dyck languages, which are languages consisting of balanced strings of parentheses, with one and multiple types of parentheses~\cite{McNaughton1967,Berstel2002}. However, a language $\poppush{L}$ cannot be viewed as such a language, as not all pop/push actions occur parenthesized: a letter $\pop{v}$ in a word $w$ in $\poppush{L}$ can be seen as an opening bracket in case it occurred because of the $\poppush{-}$ closure, but otherwise it is just a letter, and there is not necessarily a closing bracket ($\push{v}$) in $w$. Thus, the $\poppush{-}$-closure cannot be characterized using Dyck languages.

Another class of related work in which Dyck languages play an important role is language graph reachability analysis, a common way to tackle certain static analyses~\cite{Reps1997}. Informally, the nodes of a graph represent various program segments, while edges capture relationships between those segments, such as program flow or data dependencies. Analysis is then formulated as a reachability question between nodes in the graph, as witnessed by paths whose labels along the edges produce a string that belongs to a language, which is often a visibly pushdown language. This reachability problem, often referred to as Dyck reachability, has been applied in many contexts including interprocedural data-flow analysis~\cite{Horwitz1995}, slicing~\cite{Reps1994}, or type-based flow analysis~\cite{Rehof2001}. %The balanced-parenthesis property of Dyck languages is ideal for expressing sensitivity of the analysis in matching calling contexts, matching field accesses of composite objects, matching pointer references and dereferences.
 Dyck reachability has been generalized to bidirected graphs (with applications in pointer analysis) and has been studied in terms of decidability and complexity~\cite{Kjelstrom2022}.

\paragraph{Interprocedural Analysis and Pushdown Systems}
Precise interprocedural dataflow analyses involve reasoning about the call stack to capture call-return behavior~\cite{Reps1995, Sagiv1996}.
Pushdown Systems are another formalism used for program analysis and dataflow analysis~\cite{Finkel1997,Bouajjani1997}, which model a control state as well as a stack. Dataflow queries are then translated to reachability queries in pushdown systems. This avoids merging data flows from states with the same program point but different calling context, resulting in a more precise analysis. Our work by contrast is focused on program equivalence, in order to enable encoding of verification queries such as slicing, translation validation, loop freedom, etc.~\cite{Foster2015}. Weigthed versions of pushdown systems have also been developed~\cite{Reps2003, Reps2005}. It would be interesting to investigate whether techniques for interprocedural dataflow analysis can be harnessed to answer network verification queries, including weighted variants.

\paragraph{Deterministic Context Free Languages}
Another relevant subclass of languages are Deterministic Context Free Languages (DCFLs), the class accepted by deterministic pushdown automata (DPDAs). DCFLs are closed under complementation, but not under union. The equivalence problem for DPDAs is decidable~\cite{Senizergues1997,Senizergues2002,Senizergues2002a}. Like VPLs, and unlike StacKAT, DCFLs have a separate input tape.

\paragraph{Deterministic One-Counter Automata}
Another relevant subclass are the languages accepted by deterministic one-counter automata (DOCA)~\cite{Valiant1975}. A DOCA is a special case of a pushdown automaton, with a single stack that can only hold one type of symbol (often interpreted as a counter), and the stack operations are restricted to either pushing or popping this symbol (or leaving the stack unchanged). However, again because \StacKAT becomes undecidable when we have a separate stack and input tape, \StacKAT languages are not the languages captured by one-counter automata.

\paragraph{Valence Automata}
A generalization of both DOCA and pushdown automata is given by the notion of \emph{valence automaton}~\cite{Buckheister2013} where transitions are labeled with elements from a monoid. A valence automaton accepts a string depending not only on reaching a final state but also on the accumulated product of the transition labels equaling the zero of the monoid. Valence automata generalize several models and can be used to study languages and automata under algebraic constraints. In particular, a DOCA is a valence automaton for the monoid of natural numbers under addition. Pushdown automata are valence automata under the monoid of stack operations. Transitions push or pop symbols on the stack, accepting if the stack is empty. Decidability and closure properties of valence automata are dependent on the specific monoid used.

For our setting, it is not clear how to adapt valence automata to \StacKAT, because \StacKAT programs can transform the stack, and in fact the way they do so is integral to their semantics, whereas (stacks encoded in) valence automata function as an additional constraint on acceptance.

\paragraph{Generalized Versions of KA}
\tobias{The discussion of KA+H is now lumped together with omega-algebras. Maybe it deserves its own paragraph. At the same time, there is a bit of overlap with the discussion section that follows, so maybe it should just be in the discussion section?}
In~\cite{Mathieu2005}, the authors investigate an extension of Kleene algebra called omega algebra with domain to model pushdown automata, and propose axioms to model the stack that are reminiscent of ours. However, they do not discuss decidability or completeness of their algebraic theory. Recent advances in completeness techniques led to the development of the Kleene Algebra with Hypotheses framework~\cite{Doumane2019,Pous2024}, which can yield completeness proofs of several Kleene Algebra extensions. However, some of the axioms we consider in \StacKAT fail the conditions of this framework, which prevents its immediate application. In particular, closure w.r.t.\ the hypotheses $\pop{v}\cdot\push{v}\leq 1$ and $\push{v}\cdot\pop{v}\leq 1$ does not preserve regularity.

\section{Future Work}

\StacKAT is an extension of \NetKAT~\cite{Anderson2014,Foster2015,Smolka2015}, but goes beyond it by adding a stack to the data model. A new symbolic decision procedure was recently developed that can efficiently decide large \NetKAT equivalence queries~\cite{Moeller2024}. We would naturally like to integrate \StacKAT within this symbolic framework. As we can encode dup in \StacKAT, one might wonder if the integration is immediate, which is however not the case. The core of the developments in~\cite{Moeller2024} rely on the fact that during the decision procedure, when trying to determine whether $e_1$ and $e_2$ are equivalent, if a dup action is encountered then the packet state is identical at the corresponding states of $e_1$ and $e_2$.
% Intuitively, it is essential for the decision procedure that the dup action forces logging all fields of the current packet and therefore the procedure only needs to track a single packet that is processed in paralel by $e_1$ and $e_2$.
This is fundamentally different in \StacKAT: if equivalent $e_1$ and $e_2$ push the same symbol onto the stack, that does not imply that the packet state is also the same at that moment. In other words, \StacKAT programs $e_1$ and $e_2$ might be equivalent even though the packet state they are in is different. This behavioral difference makes the integration with the framework of~\cite{Moeller2024} challenging, and is left as future work.
Another promising direction for future work is to determine the \emph{active domain} of a \StacKAT program, similar to APKeep~\cite{Zhang2020}, grouping packet values together into equivalence classes, and then reducing the packet space enumeration to the equivalence classes.

\begin{acks}
  % T.~Kappé was partially supported by the European Union's Horizon 2020 research and innovation programme under grant no.\ 101027412 (VERLAN), and partially by the Dutch research council (NWO) under grant no.\ VI.Veni.232.286 (ChEOpS).
  % T.~Kappé was partially supported by the , and partially by the Dutch research council (NWO) under grant no.\ VI.Veni.232.286 (ChEOpS).
  This work was supported in part by ONR grant N68335-22-C-0411 and DARPA grant W912CG-23-C-0032, the Dutch research council (NWO) under grant no.\ VI.Veni.232.286 (ChEOpS), the European Union's Horizon 2020 research and innovation programme under grant no.\ 101027412 (VERLAN), a Royal Society Wolfson fellowship, and gifts from Google, InfoSys, and the VMware University Research Fund.
  We thank the EPFL DCSL group for providing a welcoming and supportive environment and Audrey Yuan for contributing early ideas about adding stacks to NetKAT\@.
\end{acks}

\bibliographystyle{ACM-Reference-Format}
\bibliography{bibliography}
\ifproofatend%

\appendix

\section{Proofs}
\tobias{On page 24, I am not sure about the proof of Lemma 3.3. Why does it show the forward implication? I think what you really need is that zip is a partial injection?}
\tobias{Is it our intention to fill out the proofs of the lemmas that are now stated without proof in the appendix?}

\printProofs%
\fi

\end{document}